\documentclass{article}

     \PassOptionsToPackage{numbers, compress}{natbib}
\usepackage[final]{neurips_2023}
\usepackage{tabularx}

\usepackage[utf8]{inputenc} % allow utf-8 input
\usepackage[T1]{fontenc}    % use 8-bit T1 fonts
\usepackage{hyperref}       % hyperlinks
\usepackage{url}            % simple URL typesetting
\usepackage{booktabs}       % professional-quality tables
\usepackage{amsfonts}       % blackboard math symbols
\usepackage{nicefrac}       % compact symbols for 1/2, etc.
\usepackage{microtype}      % microtypography
\usepackage{xcolor}         % colors
\usepackage{amsthm}

\RequirePackage{epsfig}
\RequirePackage{amssymb}

\RequirePackage{graphicx}

\if@usehyper
\RequirePackage[colorlinks=false,allbordercolors={1 1 1}]{hyperref}
\fi

\newtheorem{theorem}{Theorem}
\newtheorem{lemma}{Lemma}[section]
\newtheorem{corollary}{Corollary}[section]
\newtheorem{proposition}{Proposition}
\newtheorem{definition}{Definition}[section]

\newtheorem{example}{Example}
\newenvironment{remark}{\noindent{\bf Remark}
  \hspace*{1em}}{\bigskip}

\usepackage{tikz}
\usetikzlibrary{shapes.geometric, arrows}

\usepackage{subcaption}
\usepackage{mathtools}

\usepackage{overpic}
\usepackage{enumerate}
\usepackage{amssymb}
\usepackage{amsbsy}
\usepackage{amsmath}
\usepackage{amsfonts}
\usepackage{latexsym}

\usepackage{tabularx}
\makeatletter
\def\hlinewd#1{%
\noalign{\ifnum0=`}\fi\hrule \@height #1 %
\futurelet\reserved@a\@xhline}
\makeatother

\RequirePackage{algorithm}
\usepackage[noend]{algpseudocode}

\usepackage{ifthen}
\definecolor{lightgray}{rgb}{.97,.97,.97}

\newcounter{ouralgorithm}

\newcommand{\algbox}[2]{%
  \refstepcounter{ouralgorithm}
  \begin{center}
    \fcolorbox{black}{lightgray}{%
      \begin{minipage}{.9\textwidth}
        \textbf{Algorithm \theouralgorithm:} {#1}

        \vspace{-.75em} \hrulefill \vspace{0em}

        {#2}
      \end{minipage}
    }\\
  \end{center}
}

\usepackage{pgf}
\usepackage{pgfplots}
\pgfplotsset{compat=1.11}
\usepgfplotslibrary{fillbetween}
\usetikzlibrary{intersections}
\usetikzlibrary{arrows,automata}
\usetikzlibrary{positioning}

\tikzstyle{startstop} = [rectangle, rounded corners, minimum width=3cm, minimum height=1cm,text centered, draw=black, fill=white!30]
\tikzstyle{arrow} = [thick,->,>=stealth]
\newcommand{\defeq}{\coloneqq}

\newcommand{\indic}[1]{1\!\left\{#1\right\}} % Indicator function
\newcommand{\R}{\mathbb{R}}
\newcommand{\Z}{\mathbb{Z}}
\newcommand{\mc}[1]{\mathcal{#1}}

\newcommand{\norm}[1]{\left\|{#1}\right\|} % A norm with 1 argument
\newcommand{\matrixnorm}[1]{\left|\!\left|\!\left|{#1}
  \right|\!\right|\!\right|} % Matrix norm with three bars
\newcommand{\normbigg}[1]{\bigg\|{#1}\bigg\|} % A norm with 1 argument and bigg

\newtheorem{assumption}{Assumption}

\newcommand{\wopt}{w_{opt}}

\newcommand{\E}{\mathbb{E}} % Expectation symbol
\providecommand{\argmin}{\mathop{\rm argmin}}

\providecommand{\minimize}{\mathop{\rm minimize}}

\usepackage{lastpage}

\title{The $s$-value: evaluating stability with respect to distributional shifts}

\author{%
  Suyash Gupta \\
  Department of Statistics\\
  Stanford University\\
  Stanford, CA 94305 \\
  \texttt{suyash028@gmail.com} \\
   \And
   Dominik Rothenh\"ausler \\
  Department of Statistics\\
  Stanford University\\
  Stanford, CA 94305 \\
  \texttt{rdominik@stanford.edu} \\
}

\begin{document}

\maketitle

\begin{abstract}
Common statistical measures of uncertainty such as $p$-values and confidence intervals quantify the uncertainty due to sampling, that is, the uncertainty due to not observing the full population. However, sampling is not the only source of uncertainty. In practice, distributions change between locations and across time. This makes it difficult to gather knowledge that transfers across data sets. We propose a measure of instability that quantifies the distributional instability of a statistical parameter with respect to Kullback-Leibler divergence, that is, the sensitivity of the parameter under general distributional perturbations within a Kullback-Leibler divergence ball. In addition, we quantify the instability of parameters with respect to directional or variable-specific shifts. Measuring instability with respect to directional shifts can be used to detect under which kind of distribution shifts a statistical conclusion might be reversed. We discuss how such knowledge can inform data collection for transfer learning of statistical parameters under shifted distributions. We evaluate the performance of the proposed measure on real data and show that it can elucidate the distributional instability of a parameter with respect to certain shifts and can be used to improve estimation accuracy under shifted distributions.
\end{abstract}

% -*- Mode: latex -*- %

\section{Introduction}

Test data sets collected in different locations or at different time points often are drawn from different distributions, due to changing circumstances, changes in unmeasured confounders, 
time shifts in distribution, or distributional shifts in covariates \citep{Shimodaira00, DuchiNa18,DingVa16, EsfahaniKu18}. This makes it difficult to gather knowledge that transfers across data sets.
Statistical estimands such as a regression coefficient or the average treatment effect (ATE) may vary as the underlying distribution changes and hence, statistical findings (such as that the treatment effect is positive) may not replicate across data sets \citep{BasuSuHa17, GijsbertsGrHoEi15}.

In causal inference, the rapidly growing field of sensitivity analysis \citep{CornfieldHaHaLiShWy59, Rosenbaum87,DingVa16,zhao2019sensitivity,cinelli2020making} quantifies the stability of an estimate with respect to unobserved confounding. Roughly speaking, this line of work sees stability analysis as part of uncertainty quantification. Inspired by this line of work, we aim to bring a similar type of stability analysis to a wider range of statistical procedures.

In this paper, we propose a measure of instability, called the $s$-value, to investigate the stability of a given statistical parameter with respect to a shift in the underlying distribution (Figure~\ref{fig:param-change}). The $s$-value quantifies the minimum shift in distribution required to tilt the parameter to a given value, using Kullback-Leibler divergence. We also investigate the stability of parameters with respect to directional or variable-specific shifts. The proposed measure can be used as an exploratory tool to identify the kind of distribution shift that could reverse a statistical conclusion. We further discuss how $s$-values can be used to obtain improved estimates of statistical parameters under a shifted distribution with limited information about the new distribution.

\begin{figure}
    \centering
     \begin{tikzpicture}[node distance=0.4cm]
\node (Training-dist) [startstop] {Training distribution, $P_{0}$};
\node (Test-dist) [startstop, right of=Training-dist, xshift=4cm] {Shifted distribution, $P$};
\draw [arrow] (Training-dist) -- (Test-dist);

\node (Train-theta) [startstop, below of=Training-dist, yshift=-1cm] {Parameter, $\theta(P_0)>0$};
\draw [arrow] (Training-dist) -- (Train-theta);
\node (Test-theta) [startstop, below of=Test-dist, yshift=-1cm] {Parameter, $\theta(P)<0$};
\draw [arrow] (Test-dist) -- (Test-theta);
\end{tikzpicture}
\caption{Distribution shift can change the parameter of interest.}
\label{fig:param-change}
\end{figure}

\subsection{Our contribution}

We propose a measure of instability that quantifies the sensitivity of a one-dimensional statistical parameter to changes in the underlying probability distribution. We focus on shifts in distributions that are absolutely continuous with respect to the training distribution. Let $P_0 \in \mathcal{P}$ be the training distribution on the measure space $(\mathcal{Z},\mathcal{A})$, where $\mathcal{P}$ is the set of probability measures, $Z$ is a random element of $\mathcal{Z}$, and $\theta : \mathcal{P} \mapsto \mathbb{R}$ is the one-dimensional statistical parameter of interest. We are interested in the minimum amount of shift in distribution that changes the sign of the parameter. To this end, we define the stability value ($s$-value) for $\theta$ as

\begin{equation}
\label{eqn:r-value}
s(\theta,P_0)=\sup_{P \in \mc{P}}\exp{- D_{KL}(P\parallel P_0)} \hspace{0.1in} \text{s.t.} \hspace{0.1in} \theta(P) = 0,
\end{equation}
where $D_{KL}$ is the Kullback-Leibler divergence between $P$ and $P_0$ given by
\begin{align*}
D_{KL}(P\parallel P_0)=\int \log\left(\frac{dP}{dP_0}\right) dP.
\end{align*}
We provide some more discussion on the thought process that led to the proposed definition in Appendix, Section~\ref{sec:considerations}. Note that the $s$-value lies in $[0,1]$, with values close to $1$ indicating that a small shift in distribution may alter the findings, and hence, the finding is not distributionally stable. $S$-values close to $0$ indicate that the sign of the parameter is stable under distributional changes. In Section~\ref{sec:r-value}, we discuss estimation of $s$-values for parameters that are linear in the distribution $P$. The proposed procedure can be generalized to parameters defined via risk minimization, including parameters in generalized linear models (see Appendix, Section~\ref{sec:multidim_M_est}).

Considering overall distributional shift does not give information about what kind of distribution shifts the parameter is sensitive to. Hence, we also quantify the instability of parameters with respect to shifts in the distribution of certain exogenous or endogenous variables $E$, assuming that the conditional distribution of the remaining variables given $E$ is constant. Let $E$ be a random variable taking values in the space $\mathcal{E}$. We define the directional or variable-specific $s$-value as

\begin{equation}
\label{eqn:r-condn}
s_E(\theta,P_0)=\sup_{P \in \mc{P} : P(\cdot | E=e ) = P_0(\cdot | E=e) \text{ for all } e \in \mc{E}} \exp{-D_{KL}(P\parallel P_{0})} \hspace{0.1in} \text{s.t.} \hspace{0.1in} \theta(P) = 0.
\end{equation}

where $\mc{P}$ denotes the set of probability distributions over joint random variable $(Z, E)$.
If a practitioner discovers that a parameter is sensitive with respect to changes in the distribution of a certain variable $E$, this knowledge can be used to update the parameter estimate. We discuss how our method can be used to prioritize data collection about the new distribution in Section~\ref{sec:transfer-learning} and use the same to re-estimate the parameter under shifted distribution. The proposed procedure shows promise for the task of prioritizing data collection from the new distribution in the experiments.

\section{Related work}
Quantifying the uncertainty of statistical estimators is a crucial objective in statistics, typically accomplished using classical statistical measures such as $p$-values and confidence intervals to quantify sampling uncertainty. However, these methods typically rely on strong, potentially unjustified assumptions about fixed underlying distributions, which may lead to false discoveries. To improve reliability and reproducibility in statistical estimation, \citet{YuKu20} propose the predictability, computability, and stability (PCS) framework. While they investigate the stability of data results under data and method perturbations, we focus specifically on evaluating the stability of statistical parameters under distributional shifts.

Model misspecification can result in distributional instability. \citet{BujaBeBrGePiTrZhZh19} highlight fundamental issues with model misspecification or non-linearity in linear models. They propose reinterpreting population slopes as statistical functionals of data generating distributions and develop diagnostic tests for detecting model deviations. We introduce measures to illustrate coefficient instability under various distributional shifts for parametric and semi-parametric estimators.

\citep{NamkoongMaGl22} introduce a novel framework to analyze the stability of decision policies and prediction models under distribution shifts. Central to their approach is the notion of stability, characterized as the minimal alteration in the underlying environment required to push a system's performance beyond a specified threshold. In contrast, we focus on understanding stability of parameters with respect to shift in distribution.

In machine learning, there is much work on computing which data or features contribute to a prediction, e.g.\ using Shapley values \citep{lundberg2017unified,ghorbani2019data}. In contrast, we are interested in how distributional changes in features (or covariates) lead to parameter changes. %As an example, if you want to predict whether someone has a headache you might want to take into account the situation an individual is in; but understanding whether the causal effect of taking an aspirin generalizes to a new situation is a different question.

%Model misspecification can cause distributional instability, which can be assessed using classical methods like the Ramsey Regression Equation Specification Error Test (RESET) test \citep{Ramsey69} or diagnostic plots like the Tukey-Anscombe plot of residuals against fitted values. \citet{BujaBeBrGePiTrZhZh19} highlight fundamental issues with model misspecification or non-linearity in linear models and propose reinterpreting population slopes as statistical functionals of data generating distributions. They also develop diagnostic tests for detecting model deviations. We introduce measures to illustrate coefficient instability under various distributional shifts for a wide range of parametric and semi-parametric estimators.

Sensitivity analyses in the causal inference literature aim to investigate the stability of causal estimates with respect to unmeasured confounding \citep{CornfieldHaHaLiShWy59, Rosenbaum87,DingVa16,zhao2019sensitivity,cinelli2020making}. Our proposal can be seen as a version of sensitivity analysis for general estimands where we evaluate both the stability of an estimand with respect to the overall shift in distribution and the stability with respect to directional distribution shifts. 

Classical robust statistics \citep{Huber81} addresses robustness against contaminations and outliers using measures like leverage scores and influence functions to construct estimators that are not unduly influenced by such outliers. Influence functions play an important role in this work, since it corresponds to the functional derivative of a parameter with respect to the distribution. Different from classical robust statistics, in our case the perturbation is not a contamination but corresponds to an actual change in the underlying population. Rather than robustifying estimators, our aim is to equip practitioners with tools to detect sources of instability and facilitate the transfer of estimators across different settings. 

%The classical literature on robust statistics \citep{Huber81} has primarily focused on measuring robustness in the presence of contaminations and outliers, using measures such as leverage scores and influence functions, and constructing estimators that are not unduly influenced by such outliers. However, in our context, the distribution shift is induced by changes in the underlying population, rather than by the presence of outliers. Thus, estimators that are robust in the classical sense may still be unstable under distribution shift. Instead of robustifying estimators, our goal is to provide practitioners with tools to identify potential sources of instability and help them transfer estimators across different settings.

There has been a resurgence of research addressing the challenges posed by distributional shifts. This research has mostly focused on building distributionally robust estimators where more weight is given to the outliers by considering worst-case distributional shifts in a neighborhood of the training distribution \citep{DuchiNa18, SinhaNaDu18, EsfahaniKu18, Shafieezadeh-AbadehEsKu15, JeongNa20, CauchoisGuAlDu20, CauchoisGuDu21, CauchoisGuAlDu22, DuchiGuJiSu24, Gupta22, SubbaswamyAdSa21aistats}. In contrast, we quantify the stability of potentially non-linear statistical parameters under both overall and variable-specific  distributional shifts.

There is exciting empirical work by \citet{MartinNa22} who build a library of reference stability values based on survey data sets. They recommend thresholds for s-values based on empirical investigations of how data sets change between settings. This work allows for the contextualization of distributional stability values.

Closely related to our work are empirical likelihoods \citep{Owen01}. In the empirical likelihood framework, small overall distributional tilts are used to construct $p$-values and confidence intervals for a given parameter. In our work, we use distributional shifts of various strengths to evaluate the (directional) stability of an estimand with respect to distribution shifts.

\section{S-value of the mean}
\label{sec:r-value}
In this section, we discuss estimation of the $s$-value of the mean of a one-dimensional real-valued random variable followed by some examples. Estimation of $s$-values for more general settings, including parameters defined via risk minimization, is discussed in Appendix, Section~\ref{sec:multidim_M_est}. We first focus on the special case of mean estimation as it allows us to develop procedures that will be helpful in more general settings.

\subsection{Estimation of the $s$-value}
\label{sec:r-value-mean}

 Consider a one-dimensional real-valued random variable $Z \sim P_0$, where $P_0 \in \mc{P}$. We recall from \eqref{eqn:r-value} that the $s$-value for the mean $(\mu(P_0)=\E_{P_0}[Z])$ is defined  as
\begin{align}
\label{eqn:r-mean}
    s({\mu},P_0)=\sup_{P } \exp\{-  D_{KL}(P||P_0)\} \hspace{0.1in} \text{s.t.} \hspace{0.1in} \mathbb{E}_{P}[Z]=0.
\end{align}
In words, we are interested in finding the distribution closest to our training distribution $P_0$ under which the mean of the random variable is $0$. At first sight, $s$-values might seem difficult to estimate since the supremum in equation~\ref{eqn:r-mean} is taken over the infinite-dimensional space of probability distributions $\mc{P}$. However, it turns out that the $s$-value of the mean can be obtained by solving a one-dimensional convex optimization problem.
\begin{theorem}[Theorem 5.2, \citet{DonskerVa76}]
  \label{thm:Donsker-Varadhan}
  Let $Z \sim P_{0}$ be a real-valued random variable with mean $\mu(P_0)=\mathbb{E}_{P_0}[Z]$ and finite moment generating function on $\mathbb{R}$. Then, we have
\begin{align}
\label{eqn:variational-formula}
 s({\mu},P_0)= \inf_{\lambda}  \mathbb{E}_{P_0}[e^{\lambda Z}].
\end{align}
Further, if the infimum in \eqref{eqn:variational-formula} is attained at some $\lambda^* \in \R$ then the infimum in \eqref{eqn:r-mean} is attained at some probability distribution $Q$ given by
 \begin{align*}
 dQ(z)=\frac{e^{\lambda^* z}}{\E_{P_0}[e^{\lambda^* Z}]} dP_0(z) \text{ for all } z \in \R.
 \end{align*}

\end{theorem}

We note that $M_{Z}(\lambda)=\mathbb{E}[e^{\lambda Z}]$ is the moment generating function of $Z$. Since, $M_{Z}(0)=1$, we have $s({\mu},P_0) \in [0,1]$. In practice, we only have access to finitely many realizations of the data generating distribution. Let $P_n$ be the empirical distribution of $Z_i \overset{\text{i.i.d.}}{\sim} P_0$ for $i \in [n]$, we obtain an estimator of the $s$-value via the plugin estimator
\begin{align}
\label{eqn:variational-formula-finite}
\hat s({\mu},P_n)= \inf_{\lambda}  \mathbb{E}_{P_n}[e^{\lambda Z}] = \inf_{\lambda} \frac{1}{n} \sum_{i=1}^n e^{\lambda Z_i}.
\end{align}
Using classical results for $M$-estimators (see Chapter 5 of \citet{VanDerVaart98}), we show that $\hat s({\mu},P_n)$ is consistent and asymptotically normal in Appendix \ref{sec:s-val-mean}.

\paragraph{Directional $s$-values.}

The previous form of distributional stability might be very conservative. In practice, we do not expect all aspects of distribution to change from setting to setting. To allow for a more fine-grained evaluation of stability, we also consider directional shifts, which only change certain aspects of the distribution. In the following, we will make this more precise.

 Let $P_0$ be the joint distribution of the multivariate random variable $(Z,E)$ where $Z$ takes values in $\mc{Z} \subseteq \R $ and $E$ takes values in $\mc{E} \subseteq \R^p$ for some positive integer $p$. $E$ may be an exogenous or endogenous variable. We consider a directional shift, i.e.\ a situation where the marginal distribution of $E$ may change while keeping the conditional distribution of $Z$ given $E$ constant. To be more precise, we seek to estimate
 \begin{equation}
  \label{eqn:r-condn2}
    s_E(\theta,P_0)=\sup_{P \in \mc{P} :  P(\cdot | E=e ) = P_0(\cdot | E=e) \text{ for all } e \in \mc{E}} \exp\{-D_{KL}(P||P_{0})\} \hspace{0.1in} \text{s.t.} \hspace{0.1in} \theta(P) = 0.
  \end{equation}
% As above, from this characterization it is not immediate how to estimate $s_E$ since we have to deal with an infinite-dimensional optimization over probability measures $P \in \mc{P}$. 
We next show that $s_E$ is a solution to a one-dimensional convex optimization problem.
 The proof of the following result can be found in Appendix~\ref{sec:proof-cond}.

\begin{theorem}
\label{thm:Donsker-Varadhan-condn}
Let $P_0$ be the joint distribution function of the random variable $(Z,E)$ taking values in $ \mc{Z}\times\mc{E}$ with $\mu=\mathbb{E}_{P_0}[Z]$ and finite moment generating function on $\mathbb{R}$. Then,
\begin{align}
\label{eqn:variational-formula-cond}
 s_E(\mu,P_0)= \inf_{\lambda}  \mathbb{E}_{P_0}[e^{\lambda \mathbb{E}_{P_0}[Z \mid E]}].
\end{align}

Further, if the infimum in \eqref{eqn:variational-formula-cond} is attained at some $\lambda^* \in \R$ then the infimum in \eqref{eqn:r-condn} is attained at some probability distribution $Q$ given by
 \begin{align*}
 dQ(z,e)=\frac{e^{\lambda^* \mathbb{E}_{P_0}[Z \mid E=e]}}{\E_{P_0}[e^{\lambda^* \mathbb{E}_{P_0}[Z \mid E]}]} dP_0(z,e) \text{ for all } (z,e) \in \mc{Z}\times \mc{E}.
 \end{align*}
\end{theorem}
This result allows us to estimate the directional $s$-value.  Let $\hat f_n(E)$ be an estimator of $\mathbb{E}[Z \mid E]$. Then, we can define a plug-in estimator by setting
\begin{align}
\label{eqn:variational-formula-cond-finite}
 \hat s_E(\mu,P_n)= \inf_{\lambda}  \frac{1}{n} \sum_{i=1}^n e^{\lambda \hat f_n(E_i)}.
\end{align}

We prove the consistency of $\hat s_E(\mu,P_n)$ in Appendix, Section~\ref{sec:s-val-mean}. In the Appendix, we also discuss how to form a de-biased estimator of the directional $s$-value that is asymptotically normal.

\subsection{Examples}

\begin{example}[Distribution with positive support]
If $Z$ is a random variable that has positive support with probability 1, then $s({\mu},P_0)=0$, which reflects the fact that for any distribution shift within the KL-divergence ball, we will always have a positive mean.
\end{example}

\begin{example}[Gaussian distribution]
\label{exm:gaussian}
If $Z \sim N(\mu,\sigma^2)$, then $s({\mu},P_0)=e^{-\frac{\mu^2}{2 \sigma^2}}$. 
\end{example}

Thus, in the Gaussian case the stability measure $s$ is a monotonous transformation of the signal-to-noise ratio. High signal-to-noise ratio yields lower values of $s$ indicating stronger distributional stability. 

Let us now develop some intuition for directional shifts. First, we derive conditions under which the directional stability is zero, that is, conditions under which $s_E(\mu,P_0) = 0$.

\begin{example}[Directional stability]
  Let $\mathbb{E}_{P_0}[Z|E] > 0$. Then,
\begin{align*}
 s_E(\mu,P_0)= \inf_{\lambda}  \mathbb{E}_{P_0}[e^{(\lambda \mathbb{E}_{P_0}[Z \mid E])}] = \lim_{\lambda \rightarrow - \infty}  \mathbb{E}_{P_0}[e^{(\lambda \mathbb{E}_{P_0}[Z \mid E])}] = 0.
\end{align*}
\end{example}

\begin{example}[Average treatment effect] 
\label{exm:ate}
Here we consider estimating the causal effect of a treatment via the potential outcome framework \citep{NeymanDaSp90, Rubin74}. We have a  binary treatment random variable $A \in \{0,1\}$, potential outcomes $Y(1)$ and $Y(0)$ corresponding to the potential outcome under treatment and control respectively and some covariates $X$. Under the consistency assumption, we observe $Y(1)$ if $A=1$ and $Y(0)$ if $A=0$, i.e.\ $Y = A Y(1) + (1-A) Y(0)$. One can write the average treatment effect (ATE) as
\begin{equation*}
\tau =\E_{X\sim P_X}\E[Y(1)-Y(0) \mid X]= \E_{X\sim P_X}[\mu_{(1)}(X)-\mu_{(0)}(X)],
\end{equation*}
where $\mu_{(a)}(X)=E[Y(a)\mid X]$. 
Hence, if we only consider shifts in marginal distribution of covariates $X$ keeping the conditional distribution of other variables given the covariates as fixed, we obtain $s_X$-values as above with $Z = \mu_{(1)}(X)-\mu_{(0)}(X)$. In practice, $\mu_{(1)}(X)$ and $\mu_{(0)}(X)$ are often unknown. We can use plug-in estimators $\hat \mu_{(1)}(X)$ and $\hat \mu_{(0)}(X)$ to form the estimator 
\begin{equation*}
  \hat s_X(\tau,P_0) = \inf_{\lambda}  \frac{1}{n} \sum_{i=1}^n e^{\lambda ( \hat \mu_{(1)}(X_i) - \hat \mu_{(0)}(X_i) ) }.
\end{equation*} Consistency of this estimator can be shown with the same technique as Lemma~\ref{lem:consistency-mean-cond} in Appendix \ref{sec:s-val-mean}.

\end{example}

%Many statistical parameters are defined via risk minimization. In general, such parameters do not have a simple representation as mean of a random variable and hence, we cannot directly use \eqref{eqn:variational-formula-finite} or \eqref{eqn:variational-formula-cond-finite} to compute the $s$-values. We discuss methods to obtain $s$-values for such parameters in the Appendix section \ref{sec:multidim_M_est}.

In statistical analysis, risk minimization is often used to define various parameters, including regression coefficients and general M-estimators. However, unlike parameters that can be represented as the mean of a random variable, these parameters lack a simple representation as they are not linear in the underlying probability distribution. This makes the optimization problem involved in obtaining $s$-values non-convex. To address this issue, we present methods for obtaining $s$-values for such parameters in Section \ref{sec:multidim_M_est} of the Appendix.

\section{ Parameter transfer using $s$-values}
\label{sec:transfer-learning}

%and \ref{sec:multidim_M_est}

In Section~\ref{sec:r-value}, we introduced $s$-values that measure the distributional stability of statistical parameters with respect to various shifts. In this section, we discuss how we can use $s$-values to guide further data collection. The above problem of re-estimating parameters under a shifted distribution is related to the transfer learning literature that overlaps with various fields including robust machine learning, causal inference, and conformal inference \citep{PanYa10, WenYuGr14, BarberCaRaTi19b}. Here, we discuss how $s$-values can guide transfer learning.

If a parameter is unstable with respect to a shift in marginal distribution of certain covariates, then knowledge about those covariates can be used to transfer parameters across distributions. As an example, assume that we have collected some data on a job program in New York. We now want to estimate how efficient this job program would be in Boston. We have not run this job program in Boston yet, so we do not know all covariates of the participants. However, we can find that the efficiency of the job program is likely unstable with respect to changes in the demographics of job seekers in Boston. How can we use this knowledge to estimate the efficiency of the job program in Boston, based on limited data about the population in Boston? In the following, we will discuss this problem in a formal framework. We discuss another important case. Researchers are often interested in a causal effect estimate for a new location. For some covariates such as age and education, partial data is available via surveys such as American National Election Study (ANES) or Cooperative Election Study (CES). Additional partial data can be cheaply obtained via Amazon Mechanical Turk. However, some covariates are hard to collect, since they require running a study in the new location. Our numerical results show that the proposed approach can help prioritize data collection. This may drastically reduce the cost compared to running full-scale replication studies.

Assume that we want to estimate a parameter $\theta(P_{\text{shift}})$ for $P_{\text{shift}} \neq P_{0}$, but we only have observations from $P_{0}$. In addition, we may be able to collect some information about $P_{\text{shift}}$, for example, observations of a subset of variables $X_S \in \mathbb{R}^{d}$. For example, one may know the age distribution of job seekers in Boston. Intuitively, we'd like to re-weight $P_{0}$ so that the distribution of age matches the distribution of age in Boston. However, there may be infinitely many choices of weights. These different choices of weights will correspond to different values of $\theta$. Thus, in practice, it is crucial to use a form of regularization when finding a re-weighted distribution ($P_{proj}$).

We can define $P_{proj}$ as the solution of the following optimization problem:

\begin{equation*}
P_{proj} = \arg \min_{P'} D_{KL} (P' | P_0) \text{ such that } P'(X_S = \cdot) = P_{\text{shift}}(X_S = \cdot),
\end{equation*}
where $D_{KL}$ is the Kullback-Leibler divergence. The objective is similar to \citet{Hainmueller12}, where the author proposes entropy balancing to achieve covariate balance between treated and control sets for estimating the average treatment effect. This objective can be solved explicitly, leading to a covariate shift setting \citep{PanYa10}. More specifically, a short calculation shows that if the minimum is finite, then

\begin{equation*}
\mathrm{d} P_{proj}(z,x_S) = \mathrm{d} P_0(z|x_S) \mathrm{d} P_\text{shift}(x_S).
\end{equation*}

Since data collection can be costly, one would like to prioritize collecting data that is relevant for the transfer learning task. If $s_{X_S}(\theta - c ,P) = 0$ for all $c \neq \theta(P)$, then the parameter is constant under shifts in the marginal distribution of $X_S$. On the other hand, if $s_{X_S}(\theta,P) \approx 1$, then small changes in the distribution of $X_S$ might induce a large change in the parameter $\theta(\cdot)$. These heuristics motivate the following approach:

\begin{enumerate}

  \item Find variables $X_{S}$ with respect to which the parameter of interest is most sensitive to as determined by directional $s$-values. Collect observations of $X_S$ under the shifted distribution. 

  \item Estimate $\theta(P_{\text{proj}})$.
\end{enumerate}
There are several existing approaches to deal with part 2. In particular, it is possible to estimate $\theta(P_{\text{proj}})$ for a large range of estimands $\theta(\bullet)$, with asymptotically normal and efficient estimators \citep{robins1994estimation,liu2020doubly,jin2023tailored}. We can leverage these existing estimators in our workflow so that we have statistical guarantees for all steps in this pipeline. 

We will investigate the empirical performance of this two-stage approach in Section~\ref{sec:experiments}.

\newcommand{\indep}{\perp \!\!\! \perp}

\section{Experiments}\label{sec:experiments}

%In Section~\ref{exm:anscombe}, we used Anscombe's quartet to illustrate that $p$-values do not capture certain forms of distributional instability. %In this work, we develop measures that illustrate the distributional instability of various statistical parameters and further suggest methods to transfer parameters to a new distribution if they are found to be unstable with respect to distributional shifts. 
In this section, we consider real-world data to illustrate the effectiveness of the proposed methods in elucidating the distributional instability of various statistical procedures. In addition, we evaluate the two-stage transfer learning procedure described in Section~\ref{sec:transfer-learning}.

We note that $s$-values can be used to create sensitivity plots. Under various distribution shifts, the parameter can attain a range of values. More concretely, for different choices of $E$ and an upper bound on the distribution shift $c \in \R$ we define upper and lower bounds for parameter values as follows:
 \begin{align}\label{eq:upper-lower}
  \begin{split}
   \theta_{\text{upper-bound}} & =  \sup \theta(P) \text{ such that }   \\
       & P \in \mathcal{P}: P( \cdot | E=e)= P_0( \cdot |E=e)   \text{ for all $e \in \mathcal{E}$  and } D_{KL}(P \| P_0) \le c  \\
   \theta_{\text{lower-bound}} &=  \inf \theta(P) \text{  such that }  \\
   &  P \in \mathcal{P}: P( \cdot | E=e) = P_0( \cdot |E=e)  \text{ for all $e \in \mathcal{E}$ and } D_{KL}(P \| P_0) \le c.
  \end{split}
 \end{align}
 In experiments, we plot estimated versions of these upper and lower bounds across $c$ for different choices of the variable $E$. Note that if $\theta(P_0)>0$, then $s(\theta,P_0)$ is the minimum $c$ for which $ \theta_{\text{lower-bound}}=0$.

\subsection{National supported work demonstration data (NSW)}

Here, we analyze the stability of the average treatment effect estimator in the presence of covariate shift using the NSW dataset \citep{Lalonde86}, which consists of $n=722$ participants randomly assigned to a treatment or control group (variable $A$) in an employment program field experiment conducted between January 1976 and July 1977. Covariates $X$ include `age', `education', `black', `hispanic', `married', `nodegree', and `re75', where `re75' denotes pre-intervention earnings in 1975. The outcome variable is `re78', corresponding to post-intervention earnings in 1978. We apply augmented inverse probability weighting (AIPW) using causal forests \citep{WagerAt18} to estimate the average treatment effect, resulting in an estimate of $820$ with a standard deviation of $492$. We evaluate the performance of the two-stage transfer learning procedure presented in Section~\ref{sec:transfer-learning}.

%Here, we investigate the stability of the average treatment effect estimator in the presence of a covariate shift. We evaluate the performance of the two-stage transfer learning procedure presented in Section~\ref{sec:transfer-learning}, using the NSW dataset \citep{Lalonde86}. The dataset includes information on $n=722$ participants who were randomly assigned to a treatment or control group (variable $A$) in an employment program field experiment conducted between January 1976 and July 1977. The covariates $X$ include `age', `education', `black', `hispanic', `married', `nodegree', and `re75', where `re75' denotes pre-intervention earnings in 1975. The outcome variable is `re78', corresponding to post-intervention earnings in 1978. We estimate the average treatment effect using augmented inverse probability weighting (AIPW) implemented with causal forests \citep{WagerAt18}, resulting in an estimate of $820$ with a standard deviation of $492$.

\paragraph{Distributional stability of average treatment effect.}

We investigate the stability of the average treatment effect estimator under distributional changes. Specifically, we examine how $\mathbb{E}_{P}[\tau(X)]$ changes when there is a shift in the underlying distribution $P$ of each predictor separately. We measure the distributional stability of $\mathbb{E}_{P}[\tau(X)]$ while keeping the conditional distribution of the other variables given the predictor constant. The study uses the NSW dataset \citep{Lalonde86}, where the outcome variable is `re78', and the covariates $X$ include `age', `education', `black', `hispanic', `married', `nodegree', and `re75', where `re75' denotes pre-intervention earnings in 1975. We estimate the average treatment effect using augmented inverse probability weighting (AIPW) implemented with causal forests \citep{WagerAt18}, resulting in an estimate of $820$ with a standard deviation of $492$. Our findings show that the $s$-values of average treatment effect conditional on `age', `education', `black', `hispanic', and `re75' are non-zero. We also note that the average treatment effect is unstable with respect to changes in the marginal distribution of `age', `education', and `re75' indicating that the average treatment effect can change its sign with a shift in the marginal distribution of these covariates ($s_X>0.85$). We present the directional $s$-values in Table~\ref{table:NSW}.

\begin{table}[ht]
\label{table:NSW}
\centering
\caption{S-values for NSW data set}
\begin{tabular}{lccccccc}
\toprule
Feature & Age & Education & Black & Hispanic & Married & Nodegree & Re75 \\
\midrule
Directional s-value & 0.97 & 0.91 & 0.52 & 0.54 & 0 & 0 & 0.96 \\
\bottomrule
\end{tabular}
\label{table:NSW}
\end{table}

%
%We want to understand if the average treatment effect estimator is stable with respect to distributional changes. Let $\mc{X} \subseteq \R^p$ denote the predictor space and let $\tau = \E_{X\sim P_X}[\mu_{(1)}(X)-\mu_{(0)}(X)]$, where $\mu_{(a)}(X)=E[Y \mid X, A=a]$.  We study how $\mathbb{E}_{P}[\tau(X)]$ changes when there is a shift in the underlying distribution $P$. For brevity, we study shifts in each of the predictors separately. 
%We want to measure the distributional stability of $\mathbb{E}_{P}[\tau(X)]$ with respect to shifting in the marginal distribution of each of the predictor, while keeping the conditional distribution of the other variables given the predictor constant. Figure~\ref{fig:Lalonde-profile} shows the minimum and maximum achievable parameter within a given amount of shift in marginal distribution of each covariate (as defined in \eqref{eq:upper-lower}). We find that it is possible to change the sign of the average treatment effect by shifting the marginal distribution of `age', `education', `black', `hispanic', and `re75' ($s_X>0.85$). Thus, $s$-values conditional on the above variables are non-zero. Further, the average treatment effect is unstable with respect to changes in the marginal distribution of `age', `education', and `re75'. 

\paragraph{Parameter transfer.}

In this section, we evaluate the two-stage transfer learning procedure described in Section~\ref{sec:transfer-learning} using the DJW subset of the original Lalonde data extracted by \citet{DehejiaWa99} to generate training and test datasets with different distributions. The remaining samples are referred to as DJWC. The DJW subset includes 185 treated and 260 control observations and has additional information on pre-interventional earnings in 1974. We present the pre-intervention characteristics of the two subsets in Appendix~\ref{sec:appendix-expt} and find that they differ in distribution along several variables, which are statistically significant. This split into training and test datasets allows us to evaluate the transfer learning method in a setting with strong covariate shift. We estimate the average treatment effect separately in the two subsets using causal forest \citep{WagerAt18}. The estimate on the DJW subset was 1636.7 with a standard deviation of 668.8, while on DJWC, it was -847.5 with a standard deviation of 657.2.

Next, we obtain our training set by adding some proportion $\alpha$ of randomly chosen samples from the DJW subset to the DJWC subset, where $\alpha$ takes values in the set ${0.05,0.1,0.2,0.3}$, and use the remaining samples as the test set. We use the procedure described in Section \ref{sec:transfer-learning} to obtain a projection of the training distribution that closely approximates the test distribution. We use two transfer methods: full transfer, where we use all the covariates for the transfer, and partial transfer, where we use only the subset of covariates with which the ATE is most unstable, namely `age', `education', and `re75'.
We display our results in Figure \ref{fig:Lalonde_transfer}, where we find that both transfer methods lead to lower ATE estimation error than the naive procedure. Further, we observe that there is not much gain with the full transfer method that uses all the covariates over the partial transfer method.

\begin{figure*}[!htb]
\captionsetup[subfigure]{labelformat=empty}
\centering
\begin{subfigure}{.55\textwidth}
  \centering
 \begin{overpic}[
  			   %grid, %		
  				scale=0.18]{%
     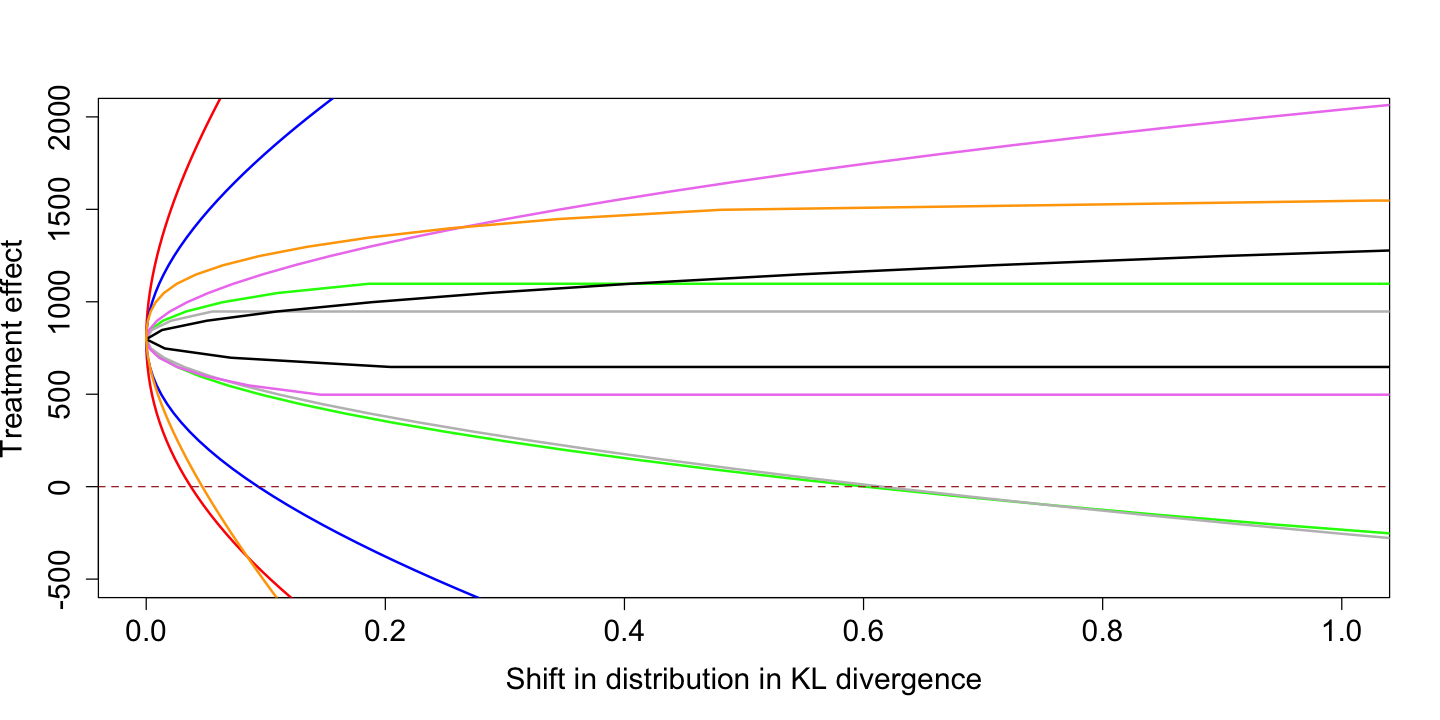}

   \put(1, -1){
      \tikz{\path[draw=white, fill=white] (0, 0) rectangle (10cm, 0.53cm)}
    }

     \put(5, 0){
        \small Shift in marginal distribution of a given covariate with respect to KL divergence }

  \end{overpic}
  \caption{ }

\end{subfigure}%
\begin{subfigure}{.58\textwidth}
  \centering
 \begin{overpic}[
  			   %grid, %		
  				scale=0.3]{%
     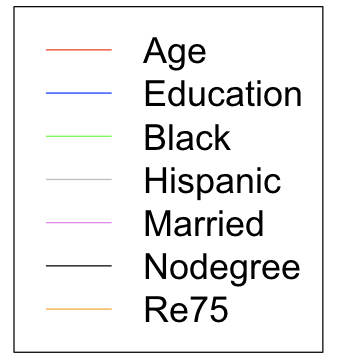}
    % Put white box over vertical axis label in upper plot

  \end{overpic}
  \caption{ }
\end{subfigure}

\caption{The plot shows the estimated minimum and maximum value   of the average treatment effect for NSW data achievable when allowing for distribution shift in some covariate (cf. equation~\eqref{eq:upper-lower}).}
\label{fig:Lalonde-profile}
\end{figure*}

\begin{figure*}[!htb]
%\captionsetup[subfigure]{labelformat=empty}
%\centering
%\begin{subfigure}{.5\textwidth}
  \centering
 \begin{overpic}[
  			   %grid, %		
  				scale=0.115]{%
     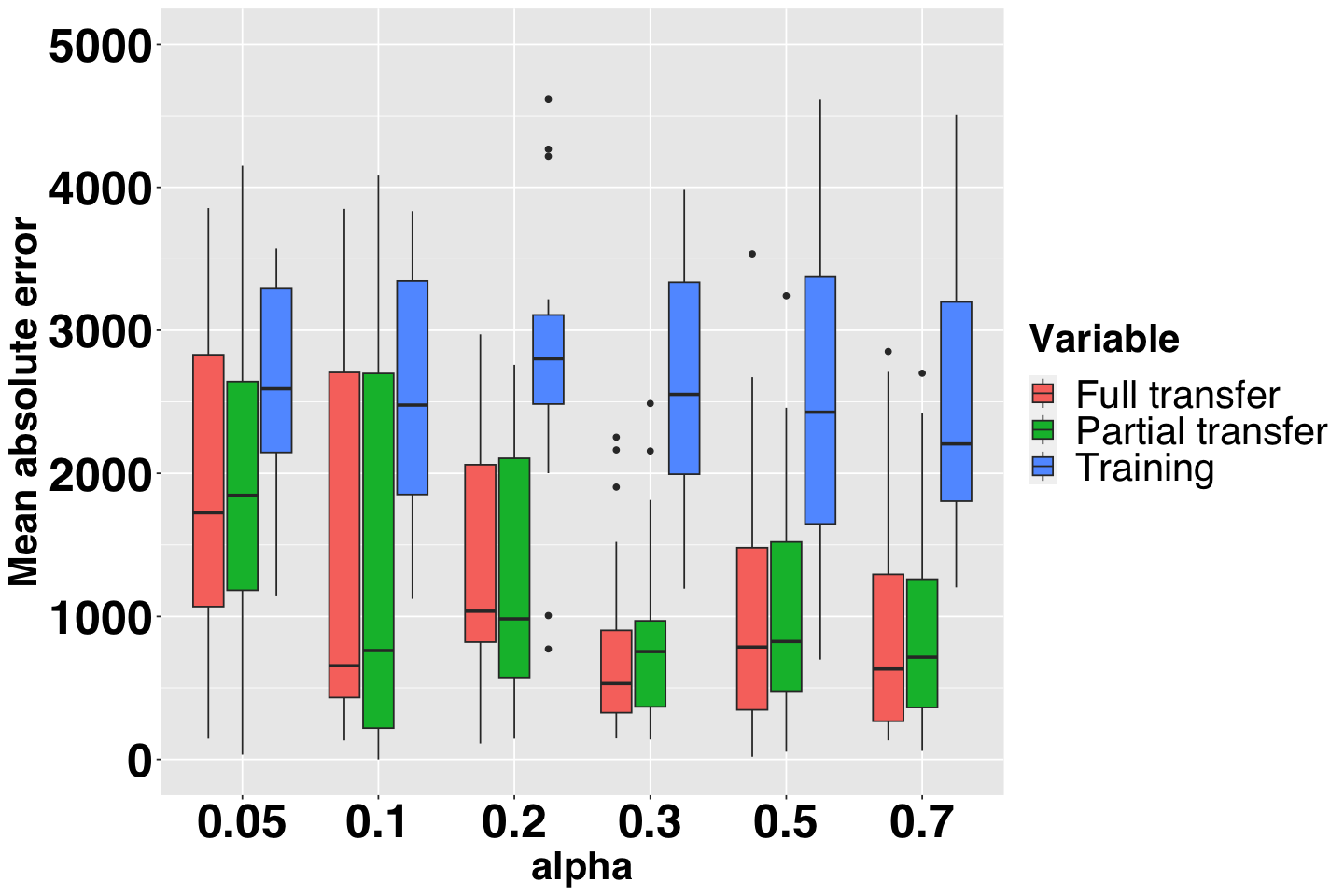}
   \put(-1.5,10){
      \tikz{\path[draw=white, fill=white] (0, 0) rectangle (.2cm, 2cm)}
    }
    \put(-6.5,20){\rotatebox{90}{
      $|\tau(P_{.})-\tau(P_{\text{test}})|$}
    }
    % Put white box over horizontal axis label in left plot
    \put(15, -1.1){
      \tikz{\path[draw=white, fill=white] (0, 0) rectangle (7cm, .3cm)}
    }
    \put(35, -2){
      alpha}
  \end{overpic}
%  \caption{ }
 % \label{fig:relative}
%\end{subfigure}%
%\begin{subfigure}{.5\textwidth}
%  \centering
% \begin{overpic}[
%  			   %grid, %		
%  				scale=0.14]{%
%     figures/lalonde_transfer_diff_alpha_difference_abs.png}
%    % Put white box over vertical axis label in upper plot
% \put(-1.5,10){
%      \tikz{\path[draw=white, fill=white] (0, 0) rectangle (.2cm, 4cm)}
%    }
%    \put(-10,10){\rotatebox{90}{
%        \small          $|\tau(P_{\text{proj}})-\tau(P_{\text{test}})|-$
%        
%        
%%        $\frac{|\tau(P_{\text{Projected}})-\tau(P_{\text{test}})|}{|\tau(P_{\text{test}}|}-\frac{|\tau(P_{\text{Training}})-\tau(P_{\text{test}})|}{|\tau(P_{\text{test}}|}$
%        }
%        
%        
%    }
%    
%       \put(-3,20){\rotatebox{90}{
%        \small          $|\tau(P_{\text{train}})-\tau(P_{\text{test}})|$
%        
%        
%%        $\frac{|\tau(P_{\text{Projected}})-\tau(P_{\text{test}})|}{|\tau(P_{\text{test}}|}-\frac{|\tau(P_{\text{Training}})-\tau(P_{\text{test}})|}{|\tau(P_{\text{test}}|}$
%        }
%        
%        
%    }
%    % Put white box over horizontal axis label in left plot
%    \put(15, -1.5){
%      \tikz{\path[draw=white, fill=white] (0, 0) rectangle (7cm, .33cm)}
%    }
%    \put(45, -2){
%      alpha}
%  \end{overpic}
%  \caption{ }
%  \label{fig:relative-difference}
%\end{subfigure}

\caption{Parameter transfer on the NSW data set. The transfer procedure described in Section~\ref{sec:transfer-learning} compared to a naive procedure that uses only the training distribution, and a full transfer procedure that uses data on all covariates from the new distribution. The green, red and blue bars represent performance of transfer learning with partial, full new data, and naive method respectively. Error bars show the range of error over 20 repetitions.}
\label{fig:Lalonde_transfer}
\end{figure*}

\subsection{Wine Quality data set}
\label{sec:wine}
We evaluate the effectiveness of our method using the wine quality dataset from the UCI Machine Learning Repository \citep{CortezCeAlMaRe09, DuaGr19}. The dataset includes subgroups of red and white wines, each with $11$ chemical properties used as predictors. The response is a continuous quality assessment measured on a scale of $0$ to $10$. The dataset consists of $1599$ red wines and $4898$ white wines. We use all red wines as the training set and randomly select a proportion $\alpha$ from the white wines, where $\alpha \in \{0.01,0.05,0.1 \}$. The remaining observations are used for testing. We include a small proportion of white wines in the red wine training set to ensure that the shifted distribution is absolutely continuous with respect to the training distribution. This step is necessary to avoid making transfer learning very challenging when datasets deviate from this assumption.

\paragraph{Distributional stability of regression coefficients.}
We use directional $s$-values to assess the distributional stability of ordinary least-squares regression coefficients. 
We focus our analysis on the predictors "pH" and "density," although similar results can be obtained for other variables. Figure~\ref{fig:wine-profile} shows estimates of the minimum and maximum achievable value of a regression coefficient given a specific shift in the marginal distribution of a covariate (as defined in equation~\eqref{eq:upper-lower}). We observe that the coefficient of "pH" is unstable with respect to shifts in "fixed.acidity," "chlorides," "pH," "sulphates," and "alcohol" ($s_X>0.85$). The coefficient of "density" is unstable with respect to shifts in "volatile.acidity," "total.sulfur.dioxide," "sulphates," and "alcohol" $(s_X>0.85)$. We present the directional $s$-values in Tables~\ref{table:ph} and \ref{table:density}.

\paragraph{Parameter transfer.}
We employ \citet{jin2023tailored}'s transfer procedure to estimate the parameter under shifted distributions. This transfer procedure combines a re-weighting step with a bias-correction step for semi-parametrically efficient transfer. We compare three different estimators. The first one uses only the training distribution, the second transfers the parameter using the covariates found to be unstable in the previous step, and the third employs all covariates for transfer (full transfer). Figure~\ref{fig:wine_transfer} displays the estimation error of the parameter under the projected and training distributions. Both transfer learning methods have smaller errors than the naive estimator that uses only the training distribution. However, the full transfer method that employs all covariates does not yield much improvement over partial transfer. For $\alpha=0.01$, there is not much enhancement in the coefficient of "density," which might be due to a partial violation of the assumption that the test distribution is absolutely continuous with respect to the training distribution.

\begin{table}[ht]
\centering
\caption{Directional S-values for wine quality data set (parameter "pH")\\   \\
f.a = fixed.acidity, v.a = volatile.acidity, c.a = citric.acid, f.so2 = free.sulfur.dioxide, t.so2 = total.sulfure.dioxide, cl = chlorides, so4 = sulphates }
\begin{tabularx}{\textwidth}{l*{11}{X}}
\toprule
Feature & f.a & v.a & c.a & r.s & cl & f.so2 & t.so2 & density & pH & so4 & alcohol \\
\midrule
s-value & 0.86 & 0.81 & 0.65 & 0.83 & 0.94 & 0.55 & 0.8 & 0.83 & 0.97 & 0.97 & 0.88 \\
\bottomrule
\end{tabularx}
\label{table:ph}
\end{table}

\begin{table}[ht]
\centering
\caption{Directional S-values for wine quality data set (parameter "density")\\   \\
f.a = fixed.acidity, v.a = volatile.acidity, c.a = citric.acid, f.so2 = free.sulfur.dioxide, t.so2 = total.sulfure.dioxide, cl = chlorides, so4 = sulphates }
\begin{tabularx}{\textwidth}{l*{11}{X}}
\toprule
Feature & f.a & v.a & c.a & r.s & cl & f.so2 & t.so2 & density & pH & so4 & alcohol \\
\midrule
s-value & 0.80 & 0.94 & 0.83 & 0.81 & 0.78 & 0.81 & 0.93 & 0.84 & 0.79 & 0.9 & 0.98 \\
\bottomrule
\end{tabularx}
\label{table:density}
\end{table}

\begin{figure*}[!htb]
\captionsetup[subfigure]{labelformat=empty}
\centering
\begin{subfigure}{.55\textwidth}
  \centering
 \begin{overpic}[
  			   %grid, %		
  				scale=0.18]{%
     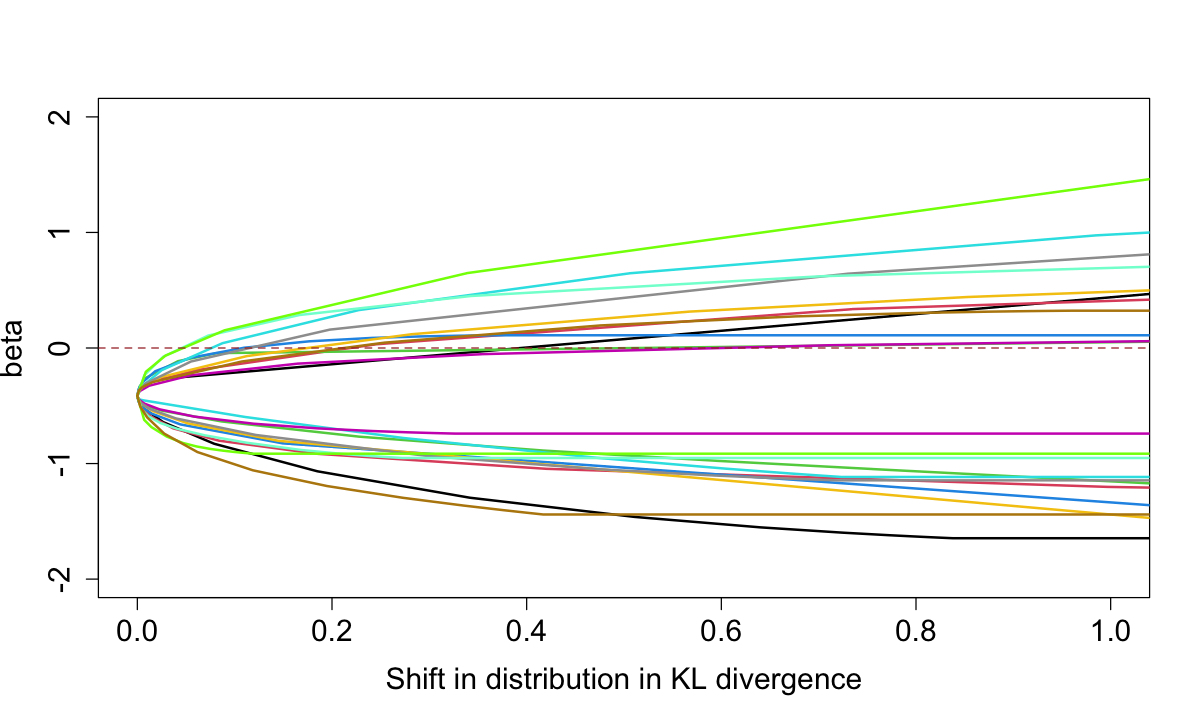}

   \put(1, -1){
      \tikz{\path[draw=white, fill=white] (0, 0) rectangle (10cm, 0.4cm)}
    }
    
         \put(5, 0){
        \small Shift in marginal distribution of a given covariate with respect to KL divergence }

    % Put white box over vertical axis label in upper plot
 \put(-1.5,10){
      \tikz{\path[draw=white, fill=white] (0, 0) rectangle (.2cm, 4cm)}
    }
   \put(-3.5,20){\rotatebox{90}{
      $\beta_{\text{pH}}$}
    }
  \end{overpic}
  \caption{ }

\end{subfigure}%
\begin{subfigure}{.62\textwidth}
  \centering
 \begin{overpic}[
  			   %grid, %		
  				scale=0.22]{%
     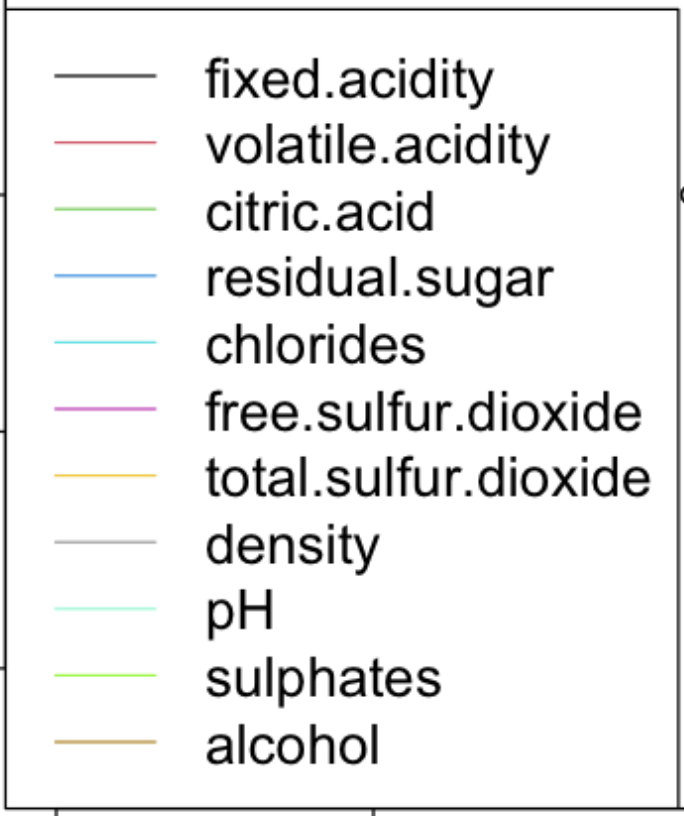}
    % Put white box over vertical axis label in upper plot

  \end{overpic}
  \caption{ }

\end{subfigure}
\begin{subfigure}{.55\textwidth}
  \centering
 \begin{overpic}[
  			   %grid, %		
  				scale=0.18]{%
     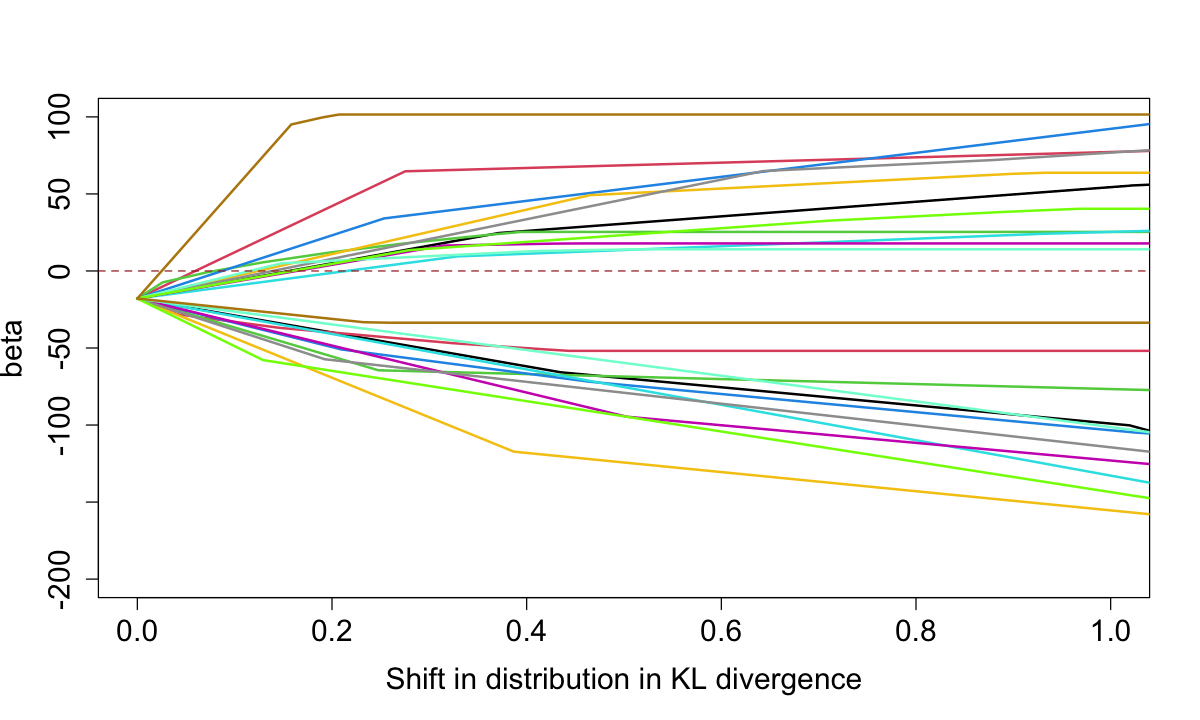}

   \put(1, -1){
      \tikz{\path[draw=white, fill=white] (0, 0) rectangle (10cm, 0.4cm)}
    }
  \put(5, 0){
        \small Shift in marginal distribution of a given covariate with respect to KL divergence }
 \put(-1.5,10){
      \tikz{\path[draw=white, fill=white] (0, 0) rectangle (.2cm, 6cm)}
    }
   \put(-3.5,20){\rotatebox{90}{
      $\beta_{\text{density}}$}
    }
  \end{overpic}
  \caption{ }

\end{subfigure}%
\begin{subfigure}{.62\textwidth}
  \centering
 \begin{overpic}[
  			   %grid, %		
  				scale=0.22]{%
     figures/wine_ph_red_train_legend_new.png}
    % Put white box over vertical axis label in upper plot

  \end{overpic}
  \caption{ }

\end{subfigure}

\caption{The plot shows the estimated minimum and maximum value of the regression coefficient for wine quality data set achievable under a distribution shift in one covariate (as defined in equation~\eqref{eq:upper-lower}).}
\label{fig:wine-profile}
\end{figure*}

\begin{figure*}[!htb]
%\captionsetup[subfigure]{labelformat=empty}
\centering
\begin{subfigure}{.5\linewidth}
  \centering
 \begin{overpic}[
  			   %grid, %		
  				scale=0.12]{%
     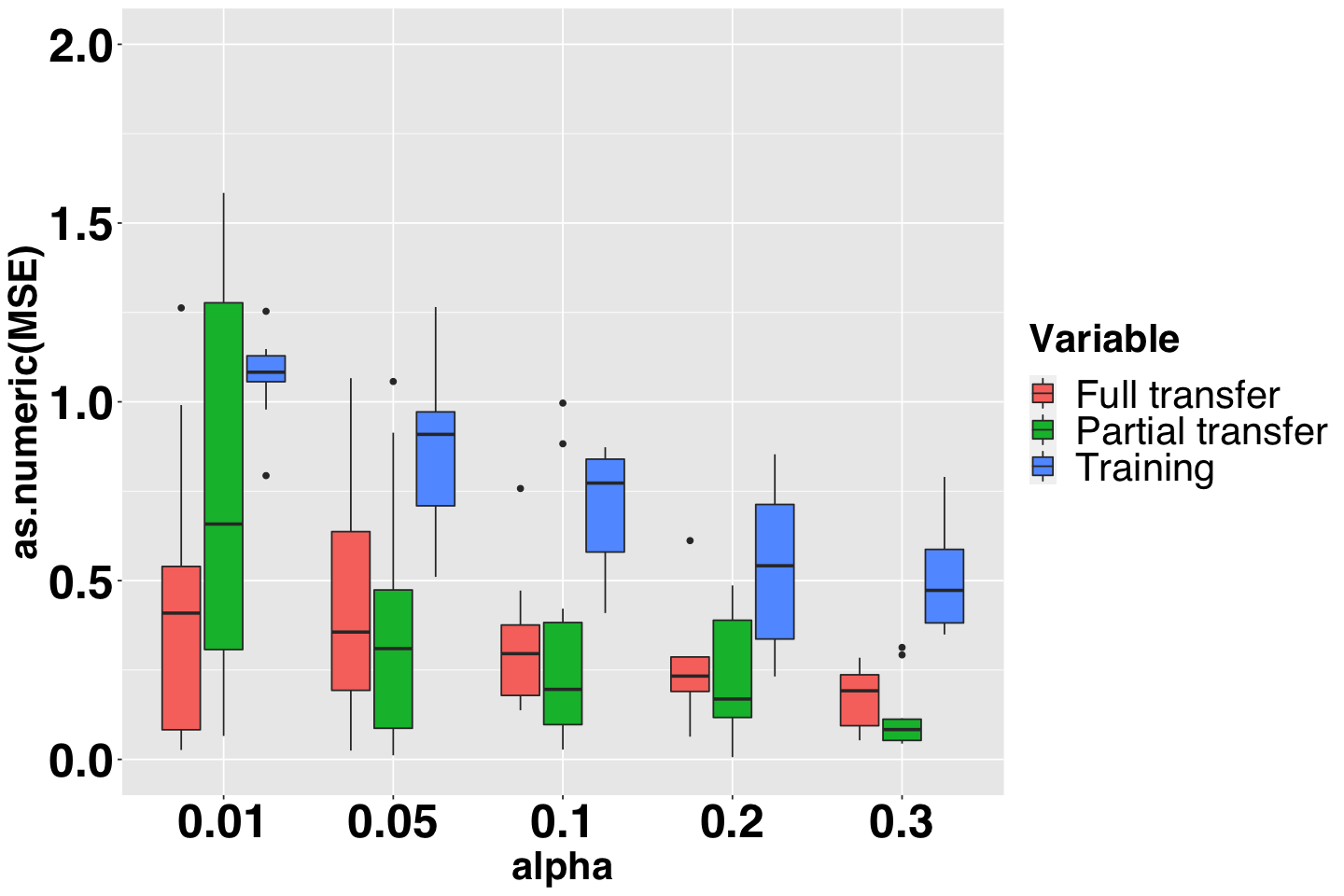}
   \put(-1.3,10){
      \tikz{\path[draw=white, fill=white] (0, 0) rectangle (.2cm, 4cm)}
    }
    \put(-2.5,10){\rotatebox{90}{
      $|\beta_{\text{pH}}(P_{.})-\beta_{\text{pH}}(P_{\text{test}})|$}
    }
    % Put white box over horizontal axis label in left plot
    \put(15, -1.5){
      \tikz{\path[draw=white, fill=white] (0, 0) rectangle (7cm, .33cm)}
    }
    \put(35, -2){
      \small alpha}
  \end{overpic}
  \caption{ }
  \label{fig:relative}
\end{subfigure}%
\begin{subfigure}{.5\textwidth}
  \centering
 \begin{overpic}[
  			   %grid, %		
  				scale=0.12]{%
     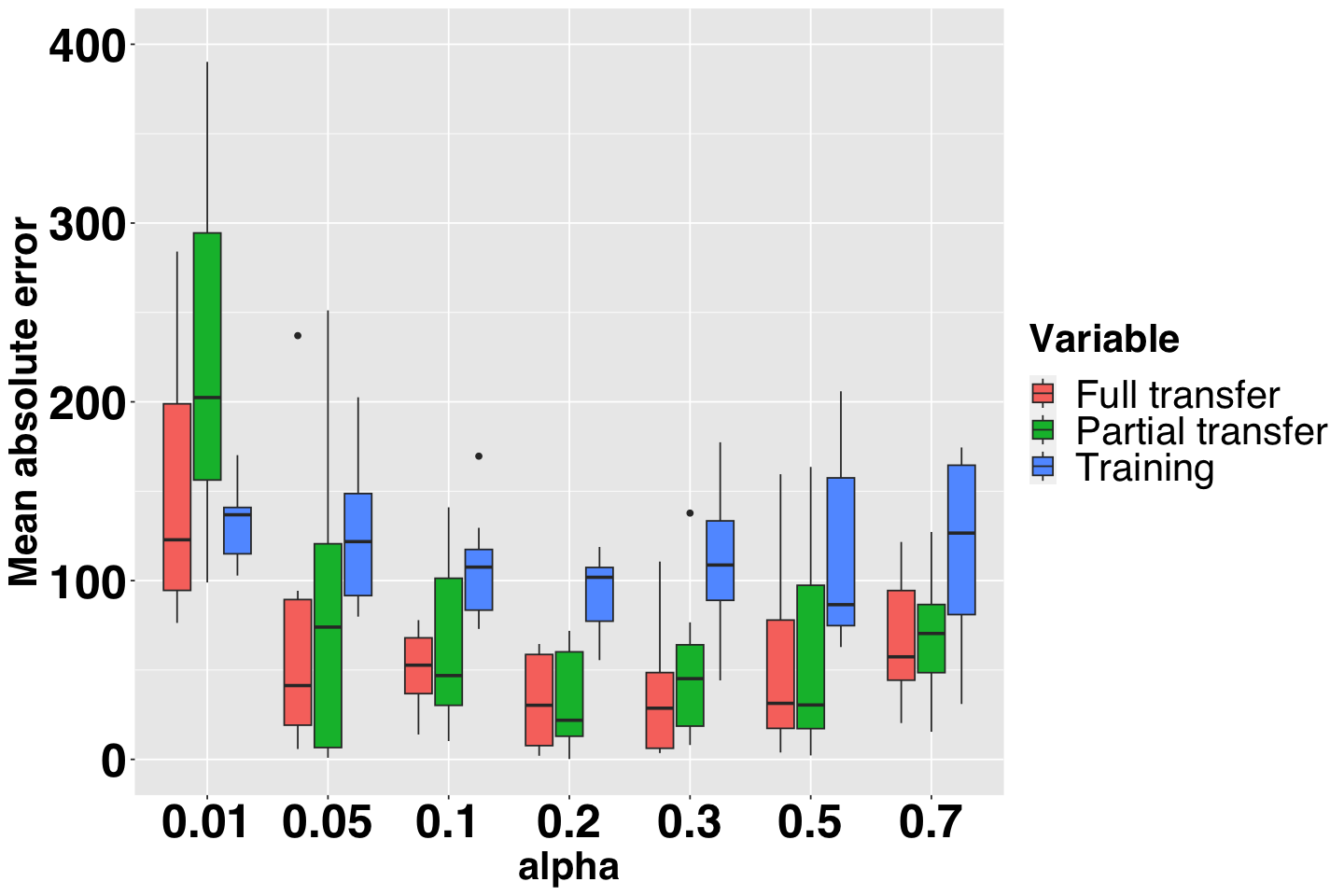}
    % Put white box over vertical axis label in upper plot
 \put(-1.3,10){
      \tikz{\path[draw=white, fill=white] (0, 0) rectangle (.2cm, 4cm)}
    }
   \put(-2.5,4){\rotatebox{90}{
      $|\beta_{\text{density}}(P_{.})-\beta_{\text{density}}(P_{\text{test}})|$}
    }
    % Put white box over horizontal axis label in left plot
    \put(15, -1.5){
      \tikz{\path[draw=white, fill=white] (0, 0) rectangle (7cm, .33cm)}
    }
    \put(35, -2){
      \small alpha}
  \end{overpic}
  \caption{ }
\end{subfigure}

\caption{ This figure shows the effectiveness of a two-stage transfer procedure for the wine quality data set. The green, red and blue bars represent performance of transfer learning with partial, full new data, and naive method respectively. Error bars show the range of error over 20 repetitions. Transfer learning outperforms the naive method in almost all cases.}
\label{fig:wine_transfer}
\end{figure*}

\section{Discussion}
\label{sec:conclusion}

The generalizability and replicability of statistical findings are crucial in scientific research. However, classical statistical measures only account for uncertainty due to sampling and not other sources of variation, such as distributional shift. Since distributions are expected to vary between settings and locations, it is essential to understand how statistical parameters are affected by such shifts to assess the stability of a finding. 

In this work, we propose stability measures to quantify the impact of distributional shifts on statistical parameters at an overall and variable-specific level, respectively, enabling a more detailed evaluation of instability. We expect that stability measures will be used in tandem with transfer learning procedures. Initially, the stability of a conclusion can be assessed with respect to various shifts, and then, once the sensitivities are determined, the data scientist can collect data from target distributions for the most sensitive covariates and update the model accordingly. As an example, researchers are often interested in a causal effect estimate for a new location. For some covariates such as age and education, partial data is available via surveys such as American National Election Study (ANES) or Cooperative Election Study (CES). Additional partial data can be cheaply obtained via Amazon Mechanical Turk. However, some covariates are hard to collect, since they require running a study in the new location. Our numerical results show that the proposed approach can help prioritize data collection.

%\begin{remark}

% We expect that $s$-values will be used to compare parameter stability under different types of distribution shift. However, some practitioners may prefer a concrete threshold to determine whether a parameter is unstable or not. These thresholds should generally depend on the generalization task and prior knowledge about expected distribution shifts. Emerging work, such as that by \citet{devaux2022quantifying}, aims to provide a quantitative answer to this question by estimating the Kullback-Leibler divergence between national surveys. This work allows for the contextualization of distributional stability values.

The proposed approach has several limitations. First, we consider worst-case shifts, which can be somewhat pessimistic for situations where the distribution is expected to change due to random perturbations in background characteristics. Modelling distribution shift as random is an attractive alternative \citep{rothenhausler2022distributionally,jeong2022calibrated}. Secondly, the Kullback-Leibler divergence only allows for certain types of distribution shift. The shifted distribution might have a different support than the training distribution. In this case, it might be more appropriate to model the changes with the total variation distance or the Wasserstein distance. Thirdly, in this paper we focus on distribution shift in the covariates $X$. There is evidence that distributions shift not only in the covariates $X$, but also in $Y|X$ \citep{jin2023diagnosing,liu2023need}. Developing new tools to better understand distributional shifts in $Y|X$ is an exciting research direction.

\newpage

%\section*{References}
\vskip 0.2in
\bibliographystyle{apalike}
\bibliography{bib}

\clearpage % Starts a new page
\pagenumbering{arabic} % Resets page numbering to 1
\newpage
\section{Appendix}
\appendix

% -*- mode: latex -*- %

%\newpage
\appendix

\section{Considerations for defining stability values}\label{sec:considerations}

In the following we want to lay out the thought process that has led us to this exact formulation of stability values. To sum it up, our choice of the $s$-value is guided by familiarity, practicality, and flexibility.

One of the advantages of using the KL divergence is that the procedure is a convex optimization problem for linear estimands. There are other considerations from a practical perspective. First, all continuously differentiable $f$-divergences have an equivalent Taylor expansion in a neighborhood around $0$. From this perspective, for small shifts it does not matter which divergence to choose. Compared to other $f$-divergences, one advantage of the KL divergence is that it is widely known in the statistics and ML communities. 

Another popular distance is the total variation distance. The KL divergence is less conservative than the total variation distance which would create very drastic distributional changes. To be more precise, for the parameter $\theta(P') = \mathbb{E}_{P'}[X]$ and all distributions $P$, we would have
\begin{equation*}
  \inf_{P'} \text{TV}(P_0,P') \text{ such that } \mathbb{E}_{P'}[X] = 0.
\end{equation*}
This infimum is zero, that is the parameter is unstable for any $P$. Thus, the stability value defined via total variation distance is too coarse in the sense that infinitely small shifts will already break common parameters of interest.  Thus, its usefulness is very limited in practice.

Comparing our choice 
\begin{equation*}
  s({\mu},P_0)=\sup_{P } \exp\{-  D_{KL}(P||P_0)\} \hspace{0.1in} \text{s.t.} \hspace{0.1in} \mathbb{E}_{P}[Z]=0,
\end{equation*}
with the potential choice
\begin{align*}
      s'({\mu},P_0)=\inf_{P }  D_{KL}(P||P_0) \hspace{0.1in} \text{s.t.} \hspace{0.1in} \mathbb{E}_{P}[Z]=0,
  \end{align*}
the latter can lead to extremely unstable optimization procedures if the argmin is far away from $P_0$. Trying to estimate $s'$ could lead to unreliable and misleading stability evaluations. Our choice of transformation guarantees that the regime where estimation of the KL divergence is unstable, leads to $s$ values that are close to each other. In other words, we "compress" regimes in which the $s$-values are hard to estimate.

\begin{figure}[!htb]
  \centering
  % To show a grid to better position drawings, uncomment grid
  \begin{overpic}[
  			   %grid, %		
  				scale=0.24]{%
     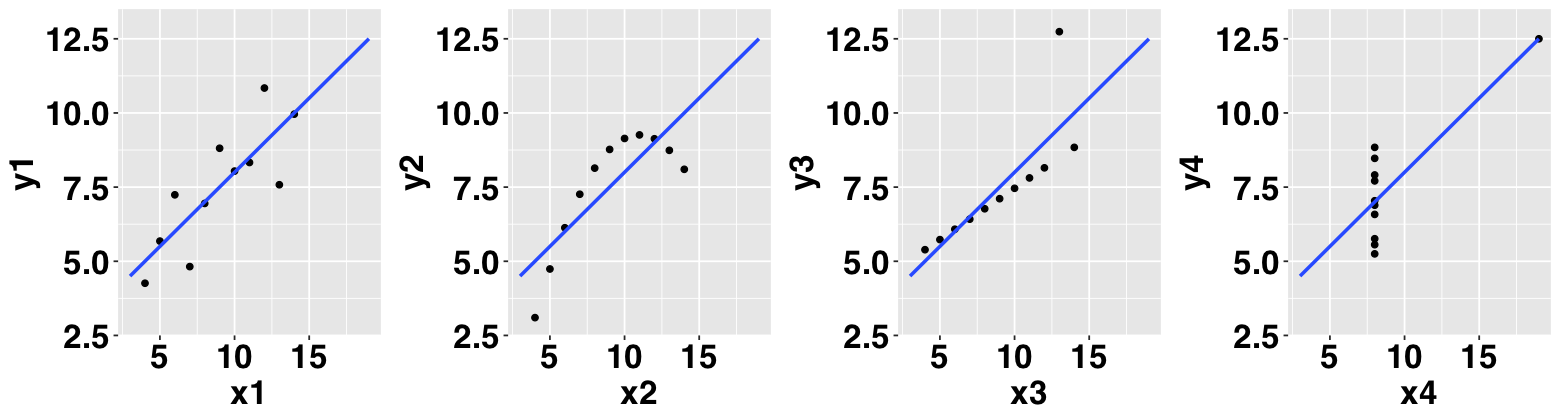}
    % Put white box over vertical axis label in left plot
  \end{overpic}
  \vspace{0.1in}
    \caption{
        Anscombe's quartet data}
\label{fig:Anscombe_quartet}
\end{figure}

\begin{table}
  \caption{This table exhibits the OLS estimate, $p$-values, general and variable-specific $s$-values of the regression coefficient for each data set in Anscombe's quartet. The $p$-values are the same for all data sets. The $s$-values vary drastically across the data sets, indicating that the parameter is more stable under distribution shift for data set 3 and data set 4 than for data set 1 and data set 2.}
  \label{tab:anscombe}
  \centering
  \begin{tabular}{lllll}
    \toprule
    \multicolumn{2}{c}{Part}                   \\
    \cmidrule(r){1-2}
 $Y=\beta_0+X\beta_1$&OLS estimate ($\beta_1$)& $p$-values &$s$&$s_X$\\
    \midrule
Set 1  & $0.5$  & $0.00217 $ &$0.465$&   $0$\\
Set 2& $0.5$  & $0.00217 $ &  $0.63$  &$0.63 $\\
Set 3& $0.5$  & $0.00217  $ &  $0$  &$0 $\\
Set 4& $0.5$  & $0.00217 $ &  $0$  &$0 $\\
    \bottomrule
  \end{tabular}
\end{table}

%\begin{table}[!htb]
%\centering
%\begin{tabular}{ |p{2.5cm}|p{4cm}|p{2cm}|p{1cm}|p{1cm}| }
%
% \hline
%$Y=\beta_0+X\beta_1$&OLS estimate ($\beta_1$)& $p$-values &$s$&$s_X$\\
% \hline
%Set 1  & $0.5$  & $0.00217 $ &$0.465$&   $0$\\
%Set 2& $0.5$  & $0.00217 $ &  $0.63$  &$0.63 $\\
%Set 3& $0.5$  & $0.00217  $ &  $0$  &$0 $\\
%Set 4& $0.5$  & $0.00217 $ &  $0$  &$0 $\\
%
% \hline
%\end{tabular}
%\caption{This table exhibits the OLS estimate, $p$-values, general and variable-specific $s$-values of the regression coefficient for each data set in Anscombe's quartet. The $p$-values are the same for all data sets. The $s$-values vary drastically across the data sets, indicating that the parameter is more stable under distribution shift for data set 3 and data set 4 than for data set 1 and data set 2.  }
%\label{tab:anscombe}
%\end{table}
\begin{figure}[!htb]
\centering
 \begin{overpic}[
  			   %grid, %		
  				scale=0.23]{%
     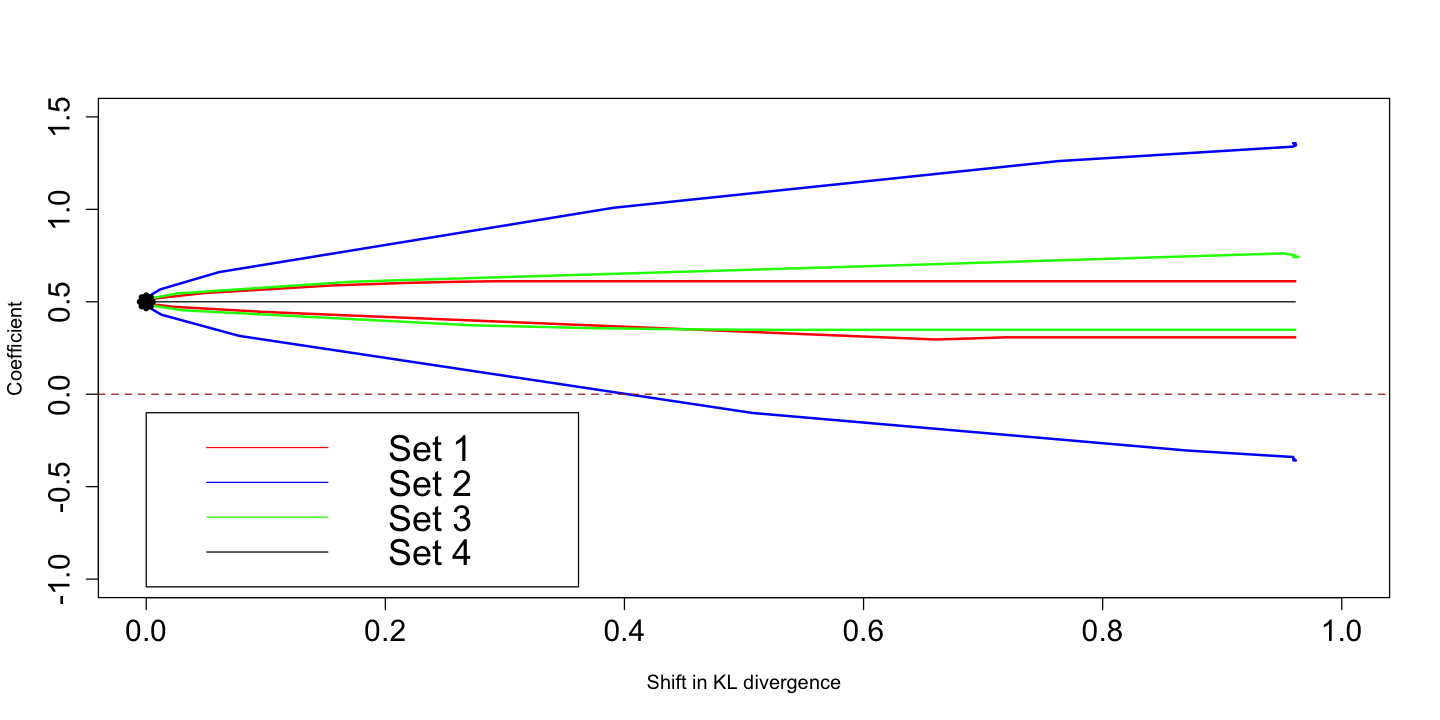}
    % Put white box over vertical axis label in upper plot
    \put(-1,0){
      \tikz{\path[draw=white, fill=white] (0, 0) rectangle (.25cm, 5cm)}
    }
    \put(-1,20){\rotatebox{90}{
        \large $\beta$}
    }
    % Put white box over horizontal axis label in upper plot
%    \put(15, -1){
%      \tikz{\path[draw=white, fill=white] (0, 0) rectangle (4cm, .2cm)}
%    }
    % Put white box over horizontal axis label in lower plot
    \put(1, -1){
      \tikz{\path[draw=white, fill=white] (0, 0) rectangle (10cm, 0.6cm)}
    }

     \put(7, 0){
        \small  Shift in marginal distribution of the covariate with respect to KL divergence }
  \end{overpic}
  \caption{The plot shows the minimum and maximum value of the regression coefficient ($\beta$) achievable by a shift in the marginal distribution of the covariate $X$. More specifically, the upper and lower bounds are estimated versions of the bounds in equation~\eqref{eq:upper-lower}.}

  \label{fig:Anscombe_profile}
\end{figure}%

\section{Example: Anscombe's quartet}
\label{exm:anscombe}

We demonstrate the usage of our method on Anscombe's quartet \citep{Anscombe73}, which comprises of four data sets that yield nearly identical OLS estimates and $p$-values (see Figure \ref{fig:Anscombe_quartet}). While $p$-values cannot unveil the difference in distributional stability of the regression coefficients among the four data sets, the proposed measure captures the stability of the regression coefficients under distribution shift.

 In Table \ref{tab:anscombe}, we display the OLS estimate, $p$-values and our $s$-values (both general and variable specific).  While the $p$-value is the same for all data sets, the $s$-values differ. The regression coefficient in set 1 has a $s$-value of $0.465$, which indicates that the regression coefficient may be null under distributional shifts. However, when considering directional shifts with $E=X$, one obtains the directional $s$-value $s_X = 0$. For Set 2, both types of $s$-values take the same non-zero value while sets 3 and 4 have $ s = s_E = 0$. Thus, the proposed stability measure coincides with the intuition that the regression coefficient is relatively stable in data sets 3 and 4 under distribution shift.

% $S$-values can also be used to create sensitivity plots. Under various distribution shifts, the linear regression coefficient can attain a range of values. More concretely, for different choices of $E$ and an upper bound on the distribution shift $x$ we define upper and lower bounds for parameter values as follows:
% \begin{align}\label{eq:upper-lower}
%  \begin{split}
%   y_{\text{upper-bound}} & =  \sup \theta(P) \text{ such that }   \\
%       & P \in \mathcal{P}: P( \cdot | E=e)= P_0[ \cdot |E=e]   \text{ for all $e \in \mathcal{E}$  and } D_{KL}(P \| P_0) \le x  \\
%   y_{\text{lower-bound}} &=  \inf \theta(P) \text{  such that }  \\
%   &  P \in \mathcal{P}: P( \cdot | E=e) = P_0( \cdot |E=e)  \text{ for all $e \in \mathcal{E}$ and } D_{KL}(P \| P_0) \le x  
%  \end{split}
% \end{align}
 In Figure \ref{fig:Anscombe_profile}, we plot estimated versions of the upper and lower bounds as defined in \eqref{eq:upper-lower} across $c$ for different choices of the variable $E$.

 In this example, distributional instability of regression coefficients mostly occurs due to model misspecification. In practice, we can test for model misspecification using classical approaches like the Ramsey Regression Equation Specification Error Test (RESET) test \citep{Ramsey69} or via diagnostic tests. However, such tests do not quantify instability in terms of distributional shifts, and distributional instability can occur even if models are well-specified. %We will discuss this in more detail in the following. 

%\paragraph{Sources of distributional instability: it is not just model misspecification}
%
%Distributional instability can occur in various situations even when a model is correctly specified. For example, the presence of exogeneous covariates that are correlated with both covariates and the outcome can change measures of association under distributional shift.  In semiparametric models, heterogeneity can induce distributional instability of estimands. As an example, if treatment effects are heterogeneous, the average treatment effect will change under distribution shift.

\section{Consistency and asymptotic normality}\label{sec:cons-norm}

\subsection{$S$-value of the mean}
\label{sec:s-val-mean}

Here we present the consistency results of the estimator of $s$-value of mean $\hat s({\mu},P_n)$ defined in equation \eqref{eqn:variational-formula-finite}.

\begin{lemma}[Consistency of $s$-value]
\label{lem:asy-norm-mean}
Let $Z \sim P_{0}$ be a real-valued non-degenerate random variable with mean $\mu(P_0)=\mathbb{E}_{P_0}[Z]$ and a finite moment generating function on $\mathbb{R}$. Then $\hat s(\mu,P_n) \overset{P}{\to} s(\mu, P_0)$ as $n \to \infty$.
\end{lemma}

We provide the proof of the above lemma in Appendix \ref{proof:asy-norm-mean}.
% \begin{lemma}[Consistency of $s$-value]
% \label{lem:asy-norm-mean}
% Let $Z \sim P_{0}$ be a real-valued random variable with mean $\mu(P_0)=\mathbb{E}_{P_0}[Z]$ and a finite moment generating function. If $\inf_{\lambda}\mathbb{E}_{P_0}[e^{\lambda Z}]$ is attained at some unique $\lambda^* \in \R$, then $\hat s(\mu,P_n) \overset{P}{\to} s(\mu, P_0)$ as $n \to \infty$.
% \end{lemma}
Now let us turn to asymptotic normality.
\begin{lemma}[Asymptotic normality of $s$-values]
  \label{lem:asy-normality-mean}
  Let $Z \sim P_{0}$ be a real-valued random variable with mean $\mu(P_0)=\mathbb{E}_{P_0}[Z]$ and a finite moment generating function on $\mathbb{R}$. Assume that $\mathbb{E}[ Z^2 e^{\lambda^*  Z}] > 0$.  Let $\hat \lambda = \arg \min \frac{1}{n} \sum_{i=1}^n e^{\lambda Z}$ and $\lambda^* = \arg \min \mathbb{E}_{P_0}[ e^{\lambda Z}]$. Then,
  \begin{equation*}
    \frac{1}{n} \sum_{i=1}^n e^{\hat \lambda Z}    - \mathbb{E}_{P_0}[e^{\lambda^* Z}] \stackrel{d}{=} \mathcal{N} \left(0,\frac{\text{Var}_{P_0}(e^{\lambda^* Z})}{n} \right) + o_P(1/n).
   \end{equation*}
  \end{lemma}
The proof for this result can be found in the Appendix, Section~\ref{sec:proof-ci}. Based on this result, one can construct confidence intervals for the $s$-value. To be more precise, for fixed $\alpha \in (0,1)$, one can define confidence intervals via $\hat s \pm z_{1-\alpha/2} \frac{\hat \sigma}{\sqrt{n}}$, where $\hat \sigma^2$ is the empirical variance of $(e^{\hat \lambda Z_i })_{i=1,\ldots,n}$ and $z_{1-\alpha/2}$ is the $1-\alpha/2$ quantile of a standard Gaussian random variable.

\textbf{Directional $s$-values}

Let us now discuss consistency of $\hat s_E(\mu,P_n)$. We make the following regularity assumption.
\begin{assumption}
\label{ass:conditional-uniform-convergence-mean}
Let $\hat f_n(\cdot)$ be an estimate of $\mathbb{E}_{P_{0}}[Z|E=\cdot]$ defined over $\mc{E}$. We assume that  $ \sup_{e \in \mc{E}} | \E_{P_{0}}[Z|E=e] - \hat f_n(e)|_{\infty} \rightarrow 0$.
\end{assumption}
\begin{lemma}[Consistency of directional $s$-value]
\label{lem:consistency-mean-cond}
Under the setting of Theorem~\ref{thm:Donsker-Varadhan-condn} and Assumption \ref{ass:conditional-uniform-convergence-mean}, we have $\hat s_E(\theta,P_n)  \overset{P}{\to}  s_E(\theta,P_0) $ as $n \to \infty$.
\end{lemma}
We present the proof of Lemma~\ref{lem:consistency-mean-cond} in Appendix~\ref{proof:consistency-mean-cond}.  To obtain asymptotically valid confidence intervals for $s$-values, one has to deal with the fact that the non-parametric estimates $\hat f_n(\cdot)$ may converge very slowly to the ground truth. To deal with this problem, we employ a debiasing technique.
\begin{lemma}[Asymptotic normality of directional $s$-values]\label{lem:asy-normality-cmean}
  Let $\hat f_n(\cdot) $ be an estimate of $f(\cdot) = \mathbb{E}_{P_0}[Z|E=\cdot]$. We assume that $\hat f_n$ is fit on a held-out portion of the data set, that is $\hat f_n(\cdot)$ is independent of $D_i, i=1,\ldots,n$. We assume that $ \sup_{e \in \mathcal{E}} | \hat f_n(e) - f(e)  | = o_P(n^{-1/4})$. Furthermore, we assume that the moment generating function of $Z$ is finite on $\mathbb{R}$ and that the matrix $\mathbb{E}_{P_0}[f(E)^2 e^{\lambda^*  f(E)}] >0 $.  Let $\hat \lambda = \arg \min \frac{1}{n} \sum_{i=1}^n e^{\lambda \hat f_n(E_i)}$ and $\lambda^* = \arg \min \mathbb{E}_{P_0}[ e^{\lambda f(E)}]$. Then,
  \begin{equation*}
   \frac{1}{n} \sum_{i=1}^n (1 + \hat \lambda Z_i - \hat \lambda \hat f_n(E_i)) e^{\hat \lambda \hat f_n(E_i)}    - \mathbb{E}_{P_0}[e^{(\lambda^*) f(E)}] \stackrel{d}{=} \mathcal{N} \left(0,\frac{\sigma_s^2}{n} \right) + o_P(1/n),
  \end{equation*}
  where
  \begin{equation*}
     \sigma_s^2 =  \text{Var}_{P_0}(e^{\lambda^* f(E)})  +  \text{Var}_{P_0}( e^{\lambda^* f(E)} \lambda^* ( Z - f(E))).
  \end{equation*}
  \end{lemma}
The proof of this result can be found in the appendix, Section~\ref{sec:proof-ci}. One can construct asymptotically valid confidence intervals based on this result, analogously as discussed after Lemma~\ref{lem:asy-normality-mean}. The de-biasing technique uses sample splitting, which reduces efficiency. Full efficiency can be obtained by using cross-fitting techniques, see for example \citet{ChernozhukovChDeDuHaNeRo16}. The debiasing technique leads to asymptotically unbiased and normal estimates, but it does come at the cost of stability. This can be easily seen from the formula: the standard approach has asymptotic variance $ \text{Var}_{P_0}(e^{\lambda^* f(E)})$, which is larger than $\sigma_s^2$ unless $Z \equiv f(E)$. For this reason and for simplicity, in the main paper we stick to the estimator in equation~\eqref{eqn:variational-formula-cond-finite}, instead of the more unstable debiased estimator. %Let us now turn to some examples. 

\section{S-values of parameters defined via risk minimization}
\label{sec:multidim_M_est}

Here, we discuss how to compute $s$-values for parameters defined via risk minimization and generalize to multi-parameter settings. Let us describe two examples where this is of interest. First, for a parameter vector $\eta$, each component might correspond to the causal effect of a subgroup. One may then ask whether there is a small distribution shift that renders all causal effects zero. The theory outlined in this section is also relevant for settings where one is interested in a single parameter in the presence of nuisance parameters. For example, in causal inference one component of the parameter vector might corresponds to the causal effect of interest, while the other components of the vector might correspond to the effect of observed confounders, which are not of interest by themselves.

We consider the following setting. Let $\Theta \subseteq \mathbb{R}^p$ be the parameter (model) space, $P_0$ be the data generating distribution (training distribution) on the measure space $(\mathcal{Z},\mathcal{A})$, $Z$ be a random element of $\mathcal{Z}$, and $L: \Theta \times \mc{Z} \rightarrow \mathbb{R}$ be a loss function, which is strictly convex and differentiable in its first argument.
Define the parameter $\theta^M(P)$ for $P \in \mc{P}$ via
\begin{equation}
\label{eq:M-est}
  \theta^M(P) = \arg \min_{\theta \in \Theta} \E_{P}[L(\theta, Z)].
\end{equation}  
Let $\ell(\theta,Z) = \partial_{\theta} L(\theta, Z)$. So far, for simplicity we have only considered $s$-values of one-dimensional parameters. In practice, we need a slightly more general notion of $s$-values that can handle $p$-dimensional parameters. For $\eta \in \mathbb{R}^{p}$, define the extended $s$-value via
\begin{equation}\label{eq:4}
   s(\theta^M-\eta,P_0)=\sup_{P \in \mc{P}}\exp\{ -D_{KL}(P||P_0)\} \hspace{0.1in} \text{s.t.} \hspace{0.1in} \theta^M(P) - \eta = 0.
\end{equation}
Choosing $\eta$ is similar to choosing a null hypothesis for significance testing in statistical decision problems. For example, in linear regression it is common to test the global null, where the regression coefficient is assumed to be zero for all components. Analogously, in this setting one might ask whether just a small distributional shift can shift all components of the parameter to zero.

Similar to the one-dimensional mean case (Section \ref{sec:r-value-mean}), the $s$-value in \eqref{eq:4} can be obtained by solving a $p$-dimensional convex optimization problem that we state in the following corollary. Its proof can be found in Appendix~\ref{proof:r-values-m}.

\begin{corollary}\label{cor:r-values-m}
Let $\ell(\eta,Z)$ have a finite moment generating function on $\mathbb{R}^p$ under $P_{0}$ for all $\eta \in \Theta$. Then, the $s$-value of the parameter $\theta^M$ as defined in \eqref{eq:M-est} is given by
\begin{equation}
  s(\theta^M-\eta,P_0) = \inf_{\lambda \in \mathbb{R}^{p}} E_{P_0}[e^{\lambda^{\intercal} \ell(\eta,Z)}].
\end{equation}
\end{corollary}
Consistency and asymptotic normality can be obtained analogously as in Lemma~\ref{lem:asy-norm-mean} and Lemma~\ref{lem:asy-normality-mean}. A general theorem is formulated in the Appendix, Section~\ref{sec:proof-ci}.
We next present some examples of parameters defined via risk minimization.
%In words, we can tilt the parameter vector $\theta^M(P_{0}) = \eta_{0} \in \mathbb{R}^{p}$ to any other vector $\theta^M(P) = \eta \in \mathbb{R}^{p}$ by solving a $p$-dimensional convex optimization problem over $\lambda$.

\begin{example}[Regression]
Let $\mc{X} \subseteq \mathbb{R}^p$ be a $p$-dimensional feature space and $\mc{Y}$ be the space of response. Let $Y \in \mc{Y}$ satisfy $Y=X\theta+\epsilon$, where $\theta \in \Theta \subset \R^p$ and $\epsilon$ is uncorrelated with $X$. Then the OLS parameter $ \theta^M(P) $ is given by
\begin{align*}
\theta^M(P)= \argmin_{\theta} \E_{P}[(Y-X\theta)^2].
\end{align*}
If $X^{\intercal}(Y-X\eta)$ has finite moment generating function on $\mathbb{R}^p$ for all $\eta \in \Theta$ then using Corollary~\ref{cor:r-values-m}, we have
\begin{equation*}
 s(\theta^M-\eta,P_0) = \inf_{\lambda \in \mathbb{R}^{p}} \E_{P_0}[e^{\lambda^{\intercal} X^{\intercal} (Y - X \eta)}].
\end{equation*}

\end{example}

\begin{example}[Generalized linear models]
\label{exm:glm}
Let $\mc{X} \subseteq \mathbb{R}^p$ be a $p$-dimensional feature space and $\mc{Y}$ be the space of response. Let $Y \in \mc{Y}$ satisfy $\E[Y \mid X]=g^{-1}(X\theta)$ where $\theta \in \Theta \subset \R^p$ and $g$ is the link function. With slight abuse of notation, let $L(Y, X \theta) $ be the negative log-likelihood function. The maximum likelihood parameter $\theta^M$ is given by
\begin{align*}
\theta^M(P)= \argmin_{\theta} \E_{P}[L(Y,X\theta)].
\end{align*}

If $X^{\intercal}\partial_2 L(Y,X\eta)$ has finite moment generating function on $\mathbb{R}^p$ for all $\eta \in \Theta$ then using Corollary~\ref{cor:r-values-m}, we have
\begin{equation*}
 s(\theta^M-\eta,P_0) = \inf_{\lambda \in \mathbb{R}^{p}} \E_{P_0}[e^{\lambda^{\intercal} X^{\intercal}\partial_{2}L(Y,X\eta)}].
\end{equation*}
\end{example}

We next characterize directional $s$-values. We define the extended directional $s$-value as
\begin{equation}\label{eq:5}
   s_E(\theta^M-\eta,P_0)=\sup_{P \in \mc{P}, P[\bullet|E] = P_{0}[\bullet|E]}\exp\{- D_{KL}(P||P_0)\} \hspace{0.1in} \text{s.t.} \hspace{0.1in} \theta^M(P) - \eta = 0.
 \end{equation}
 Similarly as above, directional $s$-values can be obtained by solving a convex optimization problem that we state in the following corollary. We present the proof in Appendix~\ref{proof:r-values-m-cond}.
\begin{corollary}[Directional shifts]\label{cor:r-values-m-cond}

Let $\ell(\eta,Z)$ have a finite moment generating function on $\mathbb{R}^p$ under $P_{0}$. Then,
\begin{equation}
  s_E(\theta^M-\eta,P_0) = \inf_{\lambda } \E_{P_0}[ e^{ \lambda^{\intercal} E_{P_0}[\ell(Z,\eta)|E]}].
\end{equation}
\end{corollary}
Consistency and asymptotic normality can be obtained analogously as in Lemma~\ref{lem:consistency-mean-cond} and Lemma~\ref{lem:asy-normality-cmean}. A general theorem is formulated in the Appendix, Section~\ref{sec:proof-ci}.

Next, we present an example from regression setting. 
\begin{example}[Regression]
\label{exm:reg-dir}
Let $\mc{X} \subseteq \mathbb{R}^p$ be a $p$-dimensional feature space and $\mc{Y}$ be the space of the response. Let $Y \in \mc{Y}$ satisfy $Y=X\theta+\epsilon$, where $\theta \in \Theta \subset \R^p$ and $\epsilon$ is independent of $X$. As above, the OLS parameter $ \theta^M(P) $ is defined as
\begin{align*}
\theta^M(P)= \argmin_{\theta} \E_{P}[(Y-X\theta)^2].
\end{align*}
If $X^{\intercal}(Y-X\eta)$ has finite moment generating function on $\mathbb{R}^p$ for all $\eta \in \Theta$, then Corollary~\ref{cor:r-values-m-cond} implies 
\begin{equation*}
 s_X(\theta^M-\eta,P_0) = \inf_{\lambda \in \mathbb{R}^{p}} \E_{P_0}[e^{\lambda^{\intercal} X^{\intercal} (\mathbb{E}_{P_0}[Y|X] - X \eta)}].
\end{equation*}
Now let us investigate the case with high directional distributional stability with respect to $E=X$. In the following, let us assume that $s_X(\theta^M-\eta,P_{0}) = 0$ for all $\eta \neq \theta(P_{0})$ and that $X$ has positive density with respect to the Lebesgue measure. By definition of the $s$-value  we have $\theta^M(P) = \theta^M(P_{0})$ for every measure $P$ that is absolutely continuous with respect to $P_{0}$ and satisfies $P(\cdot|X=x) = P_{0}(\cdot|X=x)$ for all $x \in \mc{X}$. By definition of OLS, we must have that almost surely
\begin{equation*}
  \mathbb{E}_{P_0}[ X^{\intercal} (Y - X \theta^M)|X] = 0.
\end{equation*}
Thus, almost surely,
\begin{equation*}
 X^{\intercal} (\E_{P_{0}}[Y|X] - X \theta^M)=\mathbb{E}_{P_0}[ X^{\intercal} (Y - X \theta^M)|X] = 0.
\end{equation*}  
If $X$ has a density with respect to the Lebesgue measure, then $\E_{P_{0}}[Y|X] = X \theta^M$ almost surely. Thus, directional distributional instability in linear models with respect to $E=X$ is related to whether the linear model is a good approximation of the regression surface, i.e.\ $\E_{P_{0}}[Y|X] \approx X \theta^M$. More specifically, if the linear model is a good approximation of the regression surface, directional stability is high. This is an example, where distributional instability can be induced by model misspecification. As discussed earlier in Section~\ref{exm:anscombe}, distributional instability can also be induced by other sources such as the presence of exogenous covariates that are correlated both with covariates and the outcome.

\end{example}

We can similarly obtain results for generalized linear models (see Example \ref{exm:glm}).

\subsection{S-values for a single component}

In many cases, practitioners may be interested in obtaining the $s$-value for a single component of $\theta^M \in \R^p$ instead of the entire vector $\theta^M$. For example, in causal inference, one component of $\theta^M$ can correspond to average treatment effect while other parameters may not be of scientific interest. In such settings, one might want to evaluate the stability of the parameter of interest (the average treatment effect) and not the stability of the nuisance components. Let $\theta^M_k$ be the $k$-th component of parameter vector $\theta^M \in \R^p$ for $k \in \{1, \ldots, p\}$. 

Intuitively, to obtain the $s$-value of a single parameter, we seek for the smallest possible shift in distribution that tilts the parameter to a pre-determined value, hence, we take supremum over other remaining nuissance parameters. To be more precise, using the definition of $s$-values and Corollary \ref{cor:r-values-m}, we have
\begin{align}
  \label{eq:r-val-one-M}
  \begin{split}
  s(\theta^M_k-\eta_k,P_0) &= \sup_{P \in \mc{P}} \exp\{- D_{KL}(P||P_0)\} \text{ s.t. } \eta_k(P) -\eta_k = 0  \\
  &= \sup_{\eta_1, \ldots,\eta_{k-1},\eta_{k+1}, \ldots, \eta_p} \sup_{P \in \mc{P}} \exp\{- D_{KL}(P||P_0)\} \text{ s.t. } \eta(P) -\eta = 0 \\
  &= \sup_{\eta_1, \ldots,\eta_{k-1},\eta_{k+1}, \ldots, \eta_p} \inf_{\lambda \in \R^p} \E_{P_0}[ e^{ \lambda^{\intercal} \ell(\eta,Z)}].
  \end{split}
  \end{align}
% \begin{equation}
% \label{eq:r-val-one-M}
% s(\theta^M_k-\eta_k,P_0) =\sup_{\eta_1, \ldots,\eta_{k-1},\eta_{k+1}, \ldots, \eta_p} \inf_{\lambda \in \R^p} \E_{P_0}[ e^{ \lambda^{\intercal} \ell(\eta,Z)}].
% \end{equation}

 %\citet{Owen01} considers similar quantities as \eqref{eq:4} and \eqref{eq:r-val-one-M}. The author, however, investigates small (asymptotically negligible) shifts, and the algorithms to obtain estimates of $s$-values for single components as in \eqref{eq:r-val-one-M} do not come with any convergence guarantee for large shifts. The author also does not consider any form of directional shifts. 

We next obtain a finite sample estimate of $s$-values for individual components of $\theta^M$ and show that it is consistent to the population version. Let $P_n$ be the empirical distribution of $Z_i \overset{\text{i.i.d.}}{\sim} P_0$. We propose to estimate the $s$-value via the plugin estimator

\begin{align}
\label{eq:r-val-one-M-finite}
\begin{split}
\hat s(\theta^M_k-\eta_k,P_n) &=\sup_{\eta_1, \ldots,\eta_{k-1},\eta_{k+1}, \ldots, \eta_p} \inf_{\lambda \in \R^p} \E_{P_n}[ e^{ \lambda^{\intercal} \ell(\eta,Z)}] \\
&= \sup_{\eta_1, \ldots,\eta_{k-1},\eta_{k+1}, \ldots, \eta_p} \inf_{\lambda \in \R^p} \frac{1}{n} \sum_{i=1}^n e^{ \lambda^{\intercal} \ell(\eta,Z_i)}
\end{split}
\end{align}
%\begin{equation}
%\label{eq:r-val-one-M}
%\hat s(\theta_1,P_n,\eta_1) =\max_{\eta_2, \ldots, \eta_p} \min_{\lambda \in \R^p} E_{P_n}[ e^{ \lambda^{\intercal} \ell(\eta,Z)}].
%\end{equation}
This is a challenging optimization problem. The optimization problem in \eqref{eq:r-val-one-M-finite} is non-convex and hence, finding the global optimum in practice is intractable. In Appendix \ref{sec:r-value-general}, we propose algorithms to solve the optimization problem in \eqref{eq:r-val-one-M-finite} with convergence guarantees. We also propose a simple plug-in estimate in Appendix~\ref{sec:plug-in}. 

The non-convexity of the optimization problem has important consequences for interpreting $s$-values. For the sake of simplicity, let us assume that $n$ is large, that is that estimation errors due to finite samples can be ignored. If the estimated $s$-value is large, the algorithm constructs a small distribution shift that changes the sign of the parameter. In other words, the parameter is certifiably unstable. On the other hand, if the estimated $s$-value is small, then this could be due the algorithm being stuck in a local optimum. Thus, for non-convex problems a large $s$-value is an indication of instability, but a small $s$-value cannot be necessarily interpreted as a proof of stability. This is similar to hypothesis testing in statistical inference, where the non-significance of a $p$-value cannot be interpreted as a proof of the null hypothesis.

%We discuss some algorithms in Appendix \ref{sec:r-value-general}.

We next show consistency of $\hat s(\theta^M_k-\eta_k,P_n) $ to the corresponding population stability value $s(\theta^M_k-\eta_k,P_0)$. To this end, we make the following assumption.

% in Appendix Section \ref{sec:s-val-multi}.

%
%Let us now discuss how to estimate directional $s$-values of individual components. Suppose we want to obtain the directional $s$-value of the $k$-th component of vector $\theta^M \in \R^p$. Using Corollary \ref{cor:r-values-m-cond}, the population directional $s$-value is given by
%
%\begin{equation}
%\label{eq:r-val-one-M-cond}
%s_E(\theta^M_k-\eta_k,P_0) =\sup_{\eta_1, \ldots,\eta_{k-1},\eta_{k+1}, \ldots, \eta_p} \inf_{\lambda \in \R^p} E_{P_0}[ e^{ \lambda^{\intercal}\E_{P_0}[ \ell(\eta,Z)\mid E]}].
%\end{equation}

\begin{assumption}
\label{ass:uniform_convergence}
Let $\Sigma \subset \R^p$ be a compact subset such that the map $\eta \to \ell(\eta,Z)$ is continuous on $\Sigma$ and $\E_{P_0}[\sup_{\eta \in \Sigma}e^{\lambda^{\intercal}\ell(\eta,Z)}] < \infty$ for any $\lambda \in \R^p$. 
\end{assumption}

\begin{lemma}[Consistency of $s$-value]
\label{lem:consistency-one-comp}
Let $\Sigma_k$ denote the projection of $\Sigma$ on the $k$th co-ordinate. Under Assumption~\ref{ass:uniform_convergence}, we have $\sup_{\eta_k \in \Sigma_k}|\hat s(\theta^M_k-\eta_k,P_n)  -  s(\theta^M_k-\eta_k,P_0)| \overset{P}{\to} 0 $ for $k=\{1, \ldots, p\}$  as $n \to \infty$.
\end{lemma}
We present the proof of Lemma~\ref{lem:consistency-one-comp} in Appendix~\ref{proof:consistency-one-comp}.

\begin{example}[Regression]
Let $\mc{X} \subseteq \mathbb{R}^p$ be a $p$-dimensional bounded feature space and $\mc{Y}$ be the space of the response. Let $Y \in \mc{Y}$ satisfy $Y=X^{\intercal}\beta+\epsilon$, where $\beta \in \Theta \subset \R^p$ and $\epsilon$ is independent of $X$. 
We have $L(\eta,X,Y)=\frac{1}{2}(Y-X^{\intercal}\eta)^2$, $\ell(\eta, X,Y)=-X(Y-X^{\intercal}\eta)$ and hence, $\nabla_{\eta} \ell(\eta,X,Y)=XX^T$.
Now, for a compact subset $\Sigma \subset \R$, $\E_{P_0}[\sup_{\eta \in \Sigma}e^{\lambda^{\intercal}\ell(\eta,Z)}] < \infty$ if $\epsilon$ has a finite moment generating function on $\mathbb{R}$. Invoking Lemma~\ref{lem:consistency-one-comp}, we can conclude that the estimator is consistent.

Let us now discuss how to estimate directional $s$-values of individual components. Suppose we want to obtain the directional $s$-value of the $k$-th component of vector $\theta^M \in \R^p$. Using Corollary \ref{cor:r-values-m-cond}, the population directional $s$-value is given by

\begin{equation}
\label{eq:r-val-one-M-cond}
s_E(\theta^M_k-\eta_k,P_0) =\sup_{\eta_1, \ldots,\eta_{k-1},\eta_{k+1}, \ldots, \eta_p} \inf_{\lambda \in \R^p} E_{P_0}[ e^{ \lambda^{\intercal}\E_{P_0}[ \ell(\eta,Z)\mid E]}].
\end{equation}

\end{example}

Now we propose a finite sample estimator of the directional $s$-value of individual components and show consistency. 

%in Appendix Section \ref{sec:s-val-multi}.

Let $Q_n(\eta,E)$ be a finite sample estimator of $\E_{P_0}[ \ell(\eta,Z)\mid E]$, then the finite sample plugin estimator is given by
 \begin{equation}
\label{eq:r-val-one-M-finite-cond}
\hat s_E(\theta^M_k-\eta_k,P_n) =\sup_{\eta_1, \ldots,\eta_{k-1},\eta_{k+1}, \ldots, \eta_p} \inf_{\lambda \in \R^p} E_{P_n}[ e^{ \lambda^{\intercal} Q_{n}(\eta,E)}].
\end{equation}
Again, this is a challenging optimization problem. We discuss algorithms in Appendix \ref{sec:alg-cond-s}. A simplification that allows to derive upper bounds can be found in Appendix, Section~\ref{sec:plug-in}. 
We make the following additional assumption to show consistency of $\hat s_E(\theta^M_k-\eta_k,P_n)$.
\begin{assumption}
\label{ass:conditional-uniform-convergence}
 $\sup_{\eta} \sup_{e} \norm{ E_{P_{0}}[\ell(\eta,Z)|E=e] - Q_n(\eta,e)}_{\infty} \rightarrow 0$, where $Q_n(\eta,e)$ is an estimate of $E_{P_{0}}[\ell(\eta,Z)|E=e]$.
\end{assumption}
\begin{lemma}[Consistency of directional $s$-value]
\label{lem:consistency-one-comp-cond}
Under Assumptions~\ref{ass:uniform_convergence} and \ref{ass:conditional-uniform-convergence}, we have 
\begin{align*}
\sup_{\eta_k \in \Sigma_k}|\hat s_E(\theta^M_k-\eta_k,P_n)  -  s_E(\theta^M_k-\eta_k,P_0)| \overset{P}{\to} 0  \text{ for } k=\{1, \ldots, p\}
\end{align*}
 as $n \to \infty$.
%
%$\sup_{\lambda \in \Lambda} \sup_{\eta \in \Sigma} |E_{P_n}[ e^{ \lambda^{\intercal} Q(\eta,E)}]- E_{P_0}[ e^{ \lambda^{\intercal} E[\ell(\eta,Z)|E]}]| \overset{P}{\to} 0$.
\end{lemma}
We present the proof of Lemma~\ref{lem:consistency-one-comp-cond} in Appendix~\ref{proof:consistency-one-comp-cond}. To quantify uncertainty, in this settings we find the bootstrap preferrable over asymptotic expansions (as in Section~\ref{sec:r-value}) to account for the fact that the algorithm might get stuck in a local minimum.

\section{A simple plug-in estimator}
\label{sec:plug-in}
%Here we propose a simple plug-in estimators of $s(\theta^M_k-\eta_k,P_0)$ and $s_E(\theta^M_k-\eta_k,P_0)$ as defined in \eqref{eq:r-val-one-M} and \eqref{eq:r-val-one-M-cond} respectively. 
Equation~\eqref{eq:r-val-one-M} and equation~\eqref{eq:r-val-one-M-cond} are non-convex optimization problems that are potentially difficult to solve. In practice, we can obtain a lower bound by removing the outer supremum in \eqref{eq:r-val-one-M} and \eqref{eq:r-val-one-M-cond} and using a plug-in estimator for the lower bound.

Let $\tilde{\eta}= (\hat \theta^M_1, \ldots,\hat \theta^M_{k-1}, \eta_k, \hat \theta^M_{k+1}, \ldots, \hat  \theta^M_p)$, where the  $\hat \theta^M_i$ are estimates of $\theta^M_i(P_0)$. Furthermore, let $Q_n(\eta,E)$ be an estimate of $\mathbb{E}[\ell(\eta,Z)|E]$. We can obtain plug-in estimators of $s(\theta^M_k-\eta_k,P_0)$ and $s_E(\theta^M_k-\eta_k,P_0)$ via
\begin{align*}
\hat s_{\text{plug-in}}(\theta^M_k-\eta_k,P_0)&=\inf_{\lambda \in \R^p} \E_{P_n}[ e^{ \lambda^{\intercal} \ell(\tilde{\eta},Z)}]   \\
 &= \inf_{\lambda \in \R^p} \frac{1}{n} \sum_{i=1}^n e^{ \lambda^{\intercal} \ell(\tilde{\eta},Z_i)} \text{ and } \\
 \hat s_{E,\text{plug-in}}(\theta^M_k-\eta_k,P_0)&= \inf_{\lambda \in \R^p} E_{P_n}[ e^{ \lambda^{\intercal}\E_{P_0}[ \ell(\tilde{\eta},Z)\mid E]}] \\
 & = \inf_{\lambda \in \R^p} \frac{1}{n} \sum_{i=1}^n e^{ \lambda^{\intercal}  Q_n(\tilde{\eta},E_i)}.
\end{align*}
Clearly, the objective functions are convex and hence the optimization problem is easily solvable. Large plug-in estimate certify instability of parameters. However, since these plug-in estimators are based on lower bounds of $s(\theta^M_k-\eta_k,P_0)$ and $s_E(\theta^M_k-\eta_k,P_0)$, estimates close to zero do not certify stability. Overall, the plug-in estimator can be used as a first check to evaluate distributional instability of a parameter.

\section{S-values of general estimands}
\label{sec:r-value-general}

Here, we are interested in obtaining $s$-values of individual components of parameters defined via risk minimization as in \eqref{eq:r-val-one-M}. The corresponding optimization problem to obtain $s$-value as in \eqref{eq:r-val-one-M} is generally non-convex. Hence, obtaining a globally optimal solution of the optimization problem is very challenging. Here, we characterize the form of a locally optimal solution of the corresponding optimization problem and give algorithms to solve such problems in Appendix~\ref{sec:alg}. Here we use the original definition of $s$-value as opposed to the form given in \eqref{eq:r-val-one-M}, that is,
\begin{equation}
\label{eqn:r-value-M-est}
   s(\theta^M_k,P_0)=\sup_{P \in \mc{P}}\exp\{- D_{KL}(P||P_0)\} \hspace{0.1in} \text{s.t.} \hspace{0.1in} \theta^M_k(P) = 0,
\end{equation}

For ease of presentation, from here on we denote the parameter of interest as $\theta$ instead of $\theta^M_k$ and consider a finite sample setting where we observe $n$ samples $\{Z_i\}_{i=1}^n \overset{\text{i.i.d.}}{\sim} P_0$ for some distribution $P_0 \in \mc{P}$. Let the empirical distribution of $\{Z_i\}_{i=1}^n $ be denoted by $P_{0,n}=\sum_{i=1}^n \frac{1}{n} \delta_i$, where $\delta_i$ is a dirac measure on $Z_i$. Let $W_n=[0,1]^n$ be $n$ dimensional unit cube and let $S_n=\{w \in \R^n : w_1+ \ldots + w_n=1, w_i \geq 0 \text{ for } i=1, \ldots, n\}$ be $n$ dimensional probability simplex. We focus on a one dimensional parameter $\theta : S_n \to \R$ where we define for $w \in S_n$, $\theta(w) $ as $\theta(\sum_{i=1}^n w_i \delta_i)$.
With a slight abuse of notation from now on, we redefine $\theta$ on the $n$ dimensional unit cube $W_n$ as $\theta(w)=\theta \left( \frac{\sum_{i=1}^{n} w_{i} \delta_{i}}{\sum_i w_i} \right)  $ for $w \in W_n$.
We recall that we want to obtain (extended) $s$-value of parameter $\theta$ given by

 %Let $C=\{p \in \R^n : p_0+ \ldots + p_n=1, p_i \geq 0 \text{ for } i=1, \ldots, n\}$ be $n$ dimensional probability simplex.
\begin{equation}
\label{eqn:r-value-finite}
   s(\theta-c,P_{0,n})=\sup_{w \in W_n}\exp\{- \sum_{i=1}^n w_i \log(n w_i)\} \hspace{0.1in} \text{s.t.} \hspace{0.1in} \theta(w) = c, \text{ } \sum_{i=1}^n w_i=1.
\end{equation}
where $c$ is a real constant.

The above optimization problem belongs to the class of general constrained minimization problems with equality constraints (see Chapter 3 of \citet{Bertsekas99}). In the following, we present necessary and sufficient conditions for a point to be a local optimum of problem in \eqref{eqn:r-value-finite}. This characterization can be used to verify if we obtained a locally optimal solution of our optimization problem \eqref{eqn:r-value-finite}. We first define a locally optimal solution to the problem in \eqref{eqn:r-value-finite}.

\begin{definition}
\label{def:local-optima}
An element $w^* \in S_n$ (the $n$ dimensional probability simplex) is said to be a locally optimal solution to problem \eqref{eqn:r-value-finite} if $\theta(w^*)=c$ and there exists a small $\epsilon>0$ such that $\sum_{i=1}^n w_i^*\log n w_i^* \leq \sum_{i=1}^n w_i\log n w_i$ for all $w \in S_n :  \theta(w)=c \text{ and }  \norm{w-w^*} <\epsilon$. 
\end{definition}
%\begin{remark}
%Note that the above definition of local optima hold similarly when the constraint in \eqref{eqn:r-value-finite} is replaced by $\theta(w)=c$ for any $c \in \R$.
%\end{remark}

We next present a necessary condition for a point to be a local optimum of \eqref{eqn:r-value-finite} that follows immediately from Proposition 3.1.1 of \citet{Bertsekas99}.

\begin{corollary}[Necessary conditions]
\label{thm:optim-result-finite}
Assume that $\theta : \mathrm{int}(W_n) \to \R$ is continuously differentiable. Let $w^* \in S_n$ be a locally optimal solution to problem \eqref{eqn:r-value-finite}, and assume that there does not exist a constant $r \in \R$ such that $\nabla_w\theta(w^*)=r (1, \ldots, 1)$. Then there exists a constant $\lambda \in \R$ such that
\begin{equation}
\label{eqn:wopt-form}
w^*_i \propto e^{\lambda \nabla_{i}\theta(w^*)}  \text{ for all } i=\{1, \ldots, n\}.
\end{equation}

\end{corollary}
We next present a sufficient condition for a point to be local optimum of \eqref{eqn:r-value-finite} that follows from Proposition 3.2.1 of \citet{Bertsekas99}. To that end, we introduce the Lagrangian function $h : \R^n\times \R \times \R \to \R$ that we define as
\begin{equation}
h(w,\delta, \mu)= \sum_{i=1}^{n} w_{i} \log \left( w_{i} \right) +\delta (\theta \left( w\right) -c)+\mu(\sum_{i=1}^n  w_i -1)\text{ for } w \in W_n \text{, and } \delta,\mu \in \R.
\end{equation}
\begin{corollary}[Second order sufficiency conditions]
Assume that $\theta  : \mathrm{int}(W_n) \to \R$ is twice continuously differentiable, and let $w^* \in W_n$ and $\delta^*,\mu^* \in \R$ satisfy
\begin{align*}
&\nabla_w h(w^*,\delta^*,\mu^*)=0, \text{ } \nabla_{\delta,\mu} h(w^*,\delta^*,\mu^*)=0, \\
& \gamma' \nabla_{ww}^2 h(w^*,\delta^*, \mu^*)\gamma>0, \text{ for all }\gamma \neq 0 \text{ with } \nabla \theta(w^*)'\gamma=0 \text{ and } \sum_{i=1}^n \gamma_i=0.
\end{align*}
Then $w^*$ is a strict local optimum of \eqref{eqn:r-value-finite}.
\end{corollary}

Based on the characterization of local optima above, we present a Majorization-Minimization based algorithm \citep{Lange13} in Appendix~\ref{sec:alg} to solve \eqref{eqn:r-value-finite} and give sufficient conditions under which the iterates of the algorithm converges to a point that satisfies the first-order necessary conditions \eqref{eqn:wopt-form}. We have similar characterization of locally optimal solution for the optimization problem involved in obtaining the directional $s$-values \eqref{eqn:r-condn} that we present in Appendix~\ref{sec:alg-cond-s-alg} along with the algorithm to solve such problems.

\subsection{Algorithms to obtain $s$-values for general estimands}\label{sec:alg}
Here we present a Majorization-Minimization (MM) based algorithm \citep{Lange13} to solve the problem in \eqref{eqn:r-value-finite} and show that it converges to a point that satisfies first-order necessary conditions \eqref{eqn:wopt-form}. We also adapt our procedure to obtain directional or variable specific $s$-value. We can use several existing algorithms to solve \eqref{eqn:r-value-finite} (see Chapter 4 of \citet{Bertsekas99}) that come with some convergence guarantees. However, the convergence guarantees of the existing algorithms typically come under the assumption that the iterates obtained by the algorithm converge (or we only have guarantees along a subsequence) whereas we present sufficient conditions under which the iterates obtained by our algorithm always converge to a point that satisfies first-order necessary conditions. Further, the existing algorithms require obtaining close approximations to the first-order stationary points of the corresponding augmented Lagrangian (for example, augmenting the objective function with a square of the parameter $\theta$ with a high penalty), however, standard approaches for obtaining first-order stationary points of such functions require slightly stronger assumptions (M-smoothness of the square of $\theta$, see Assumption \ref{ass:continuous-differentiability} and the following remark below). Further, since the constraint function involves a one-dimensional parameter $\theta$, our procedure can be efficiently adapted to obtain $s$-value over a range of constants $c$ as in equation \eqref{eqn:r-value-finite}.

To that end, we make the following smoothness assumption of our parameter $\theta$. 

\begin{assumption}
\label{ass:continuous-differentiability}
The function $\theta : \text{int}(W_n) \to \R$ is continuously differentiable and $M$ smooth for some $M \in \R$, that is, for $w, w' \in W_n$, 
\begin{equation}
\label{eqn:upper_bound-l2}
|\theta(w') -\theta(w)- \langle \nabla \theta(w), w'-w \rangle| \leq \frac{M}{2} \norm{w-w'}^2_2.
\end{equation}
\end{assumption}
Since $\ell_1$ and $\ell_2$ norms are equivalent in finite dimensional spaces, by Pinsker's inequality, we have the following relation for any $w, w' \in S_n$ for some real constant $L>0$,

\begin{equation}
\label{eqn:upper-bound}
|\theta(w')-\theta(w)- \langle \nabla \theta(w), w'-w \rangle| \leq L\sum_{i=1}^n w^{'}_i \log \frac{w^{'}_i}{w_i}.
\end{equation}
This new upper bound would help obtain a closed-form expression of update in each iteration of the algorithm (see Proposition \ref{prop:solve-upper-bound}).

\begin{remark}
In practice, the constant $L$ is often not known in which case, we need to tune it similarly as we would tune the step size in a gradient descent-based method.
\end{remark}

\begin{remark}
Although, it appears we make stronger smoothness assumptions for $\theta$ than what is needed for convergence guarantees of algorithms in Chapter 4 of \citet{Bertsekas99}, however, such assumptions are standard for convergence guarantees of gradient descent-based methods that are typically used to minimize the augmented Lagrangian at each iterate of the algorithm as needed for example, in Proposition 4.2.2 of \citet{Bertsekas99} where we need $M$-smoothness of the square of $\theta$.
\end{remark}

To obtain a solution of \eqref{eqn:r-value-finite}, we solve the Lagrangian form of the optimization problem in \eqref{eqn:r-value-finite} as given by

\begin{equation}
\label{eqn:general-optim-lag}
  \minimize_{w_{1},\ldots,w_{n}, w_{i} \ge 0, \sum w_{i} =1}  g(w)=\delta (\theta \left( w\right)-c) + \sum_{i=1}^{n} w_{i} \log \left( w_{i} \right)
\end{equation}
for any fixed $\delta$.
If there exists a $\delta$ such that the iterates obtained by our algorithm converges to some $\wopt^{\delta} \in W_n$ that satisfies $\theta(\wopt^{\delta})=c$, then we can show that $\wopt^{\delta}$ satisfies first order necessary conditions \eqref{eqn:wopt-form}. In practice, we use grid search to obtain a $\delta$ that yields $\theta(\wopt^{\delta})=c$.

Now we present the MM based algorithm to solve \eqref{eqn:general-optim-lag} for a given $\delta$ and parameters that satisfy Assumption~\ref{ass:continuous-differentiability}.
Without loss of generality, we assume that $\theta(P_{0,n})>c$ and hence, we choose $\delta>0$. First, we upper bound the objective in \eqref{eqn:general-optim-lag} using inequality \eqref{eqn:upper-bound} so that we have for $w,w' \in W_n$
\begin{align}
\label{eqn:upper-bound-alg}
g(w') \leq G_{L}(w',w) \defeq \delta \left(\theta(w) -c+ \langle \nabla \theta(w), w'-w \rangle + L\sum_{i=1}^n w_i'\log \frac{w_i'}{w_i}\right) +\sum_{i=1}^n w_i'\log w_i'.
\end{align}
Note that $G_{L}(w',w)$ is convex in $w'$ and that $g(w')=G_{L}(w',w)$ when $w'=w$. 
Our algorithm then runs iteratively where given a current solution $w^k$, we obtain the next iterate $w^{k+1}$ as $w^{k+1}=\argmin_{w: \sum_{i=1}^n w_i=1}G_{L}(w,w^k)$, which has a closed form that we present in the proposition below. We define the iteration map $M:W_n \to W_n$ as $M(w^k)=w^{k+1}$ for all $w_k \in W_n$.

\begin{proposition}
\label{prop:solve-upper-bound}
Let $w^{k+1}$ be the iterate obtained at $k$th iteration, that is, 
\begin{equation}
\label{eqn:convex-weights}
M(w^k)=w^{k+1}=\argmin_{w: \sum_{i=1}^n w_i=1} \delta \left(\theta(w^k) -c+ \langle \nabla \theta(w^k), w-w^k \rangle + L\sum_{i=1}^n w_i\log \frac{w_i}{w^k_i}\right) +\sum_{i=1}^n w_i\log w_i,
\end{equation}
then it is uniquely given by

\begin{equation}
\label{eqn:next-iter}
(M(w^k))_i = w_{i}^{k+1}   \propto e^{-\frac{\delta}{1+L\delta}  \nabla_i \theta(w^k)} (w^k_i)^{\frac{L\delta}{1+L\delta}},
\end{equation}
for all $i= \{1, \ldots, n\}$.
\end{proposition}

\subsection{Proof of Proposition~\ref{prop:solve-upper-bound}}
\label{proof:solve-upper-bound}
\begin{proof}
The optimization problem in \eqref{eqn:convex-weights} is a convex optimization problem and we obtain the solution to \eqref{eqn:convex-weights} via a Lagrange multipliers.

The Lagrangian is given by

\begin{equation}
\argmin_{w: \sum_{i=1}^n w_i=1} \delta \left( \langle \nabla \theta(w^k), w-w^k \rangle + L\sum_{i=1}^n w_i\log \frac{w_i}{w^k_i}\right) +\sum_{i=1}^n w_i\log w_i +\gamma (\sum_{i=1}^n w_i-1).
\end{equation}

Differentiating with respect to $w_i$ and setting the derivative to 0 gives,

\begin{equation}
\delta\nabla_i\theta(w^k)+\delta L\log\frac{w_i}{w^k_i}+\delta L+\log w_i+1+\gamma=0.
\end{equation}
Hence, the result follows after rearranging the terms and using the constraint $\sum_{i=1}^n w_i=1$.
\end{proof}

%%We present the proof of Proposition~\ref{prop:solve-upper-bound} in Appendix~\ref{proof:solve-upper-bound}.
%%\Dominik{In words, this means that solutions to xxx are fixed points of the optimization procedure.} 

Below we summarize our algorithm to solve \eqref{eqn:general-optim-lag}
for a fixed $\delta$.
\algbox{\label{alg:r-M}Solving \eqref{eqn:general-optim-lag} for a fixed $\delta$.
}{
  \textbf{Input:} Training distribution $P_{0,n}$, parameter $ \theta$ satisfying Assumption \ref{ass:continuous-differentiability} where without loss of generality $\theta(P_{0,n})>c$, penalty $\delta>0$, convergence tolerance $\epsilon$ .
  \newline
  \textbf{Output:} First order stationary solution of \eqref{eqn:general-optim-lag}.
  \newline
   %\textbf{Output:} A locally optimal solution to problem \eqref{eqn:general-optim}.
Set $k \leftarrow 0$, initialize $w^0$ with some $w \in W$, for example, $w^0_i=\frac{1}{n}$ for all $i =\{1, \ldots, n\}$.
\begin{enumerate}[1.]
  \item For $k \geq 0$, obtain $w^{k+1}$ as in \eqref{eqn:next-iter}.
 
  \item Set $k \leftarrow k+1.$
  \item Stop if $ g(w^{k+1})-g(w^{k}) \leq \epsilon.$
  \end{enumerate}
  Return $\wopt^{\delta}= w^{k+1}$.
}

We next present the convergence analysis of Algorithm~\ref{alg:r-M} in the following proposition that we prove in Section~\ref{proof:conv-optim}. First, we recall the definition of a stationary point of a constrained optimization problem where the constraint set is convex.
\begin{definition}
Consider the following optimization problem
\begin{equation}
\label{eqn:stat-opt}
\minimize_{x: x \in C} f(x) 
\end{equation}
where $f : \R^{p} \to \R$ is differentiable but possibly non-convex, $C \subset \R^p$ is a closed convex set.
We call $x^*$ a stationary point of \eqref{eqn:stat-opt} if and only if 
\begin{equation}
\langle\nabla  f(x^*),(x-x^*)\rangle \geq 0 \text{ for all } x \in C.
\end{equation}
\end{definition}

\begin{proposition}
\label{prop:conv-optim}
Let $\{w^{k}\}_{k \geq 1}$ be the sequence of probability distributions generated by Algorithm~\ref{alg:r-M}, which solves \eqref{eqn:general-optim-lag} for some fixed $\delta$ and convergence tolerance $\epsilon=0$. If there exists a constant $A$ such that $|\theta(w)| \leq A$ for all $w \in W_n$, the unit cube in $n$-dimension, then we have:
\begin{enumerate}
\item The sequence $\{g(w^k)\}_{k \geq 1}$ is decreasing and converges.

\item In addition if all stationary points of \eqref{eqn:general-optim-lag} are isolated, then the sequence $\{w^{k}\}_{k \geq 1}$ converges and if $\lim_{k \to \infty} w^k= w^*_{\delta} \neq (\frac{1}{n}, \ldots, \frac{1}{n})$, then $w^*_{\delta}$ satisfies first order necessary conditions \eqref{eqn:wopt-form}, where the constraint in \eqref{eqn:r-value-finite} is replaced with $\theta(w)=\theta(w^*_{\delta})$. 
\end{enumerate}

\end{proposition}

Next we use grid search to find $\delta$ (typically increase the value of $\delta$) such that $\wopt^{\delta}$ satisfies $\theta(\wopt^{\delta})=c$. 
Below we summarize the algorithm to find a solution of \eqref{eqn:r-value-finite} that satisfies first order necessary conditions \eqref{eqn:wopt-form}.
\algbox{\label{alg:r-general}$s$-value for general estimands.
}{
  \textbf{Input:} Training distribution $P_0$, parameter $ \theta$ satisfying Assumption \ref{ass:continuous-differentiability}, convergence tolerance $\epsilon$ .
  \newline
  \textbf{Output:} First order stationary point of \eqref{eqn:r-value-finite}.
  \newline
   %\textbf{Output:} A locally optimal solution to problem \eqref{eqn:general-optim}.
Set $k \leftarrow 1$, initialize $\delta_0=0$, $\delta_1=2\gamma$ for some small $\gamma>0.$ 
\begin{enumerate}[1.]
  \item  Run Algorithm~\ref{alg:r-M} with $\delta=\delta_k$ and obtain the output of Algorithm~\ref{alg:r-M} as $\wopt^{\delta_k}$.
  
  \item If $|\theta(\wopt^{\delta_k})-c| \leq \epsilon$, stop and return $s(\theta-c,P_{0,n})=e^{-\sum_{i=1}^n (\wopt^{\delta_k})_i\log {n(\wopt^{\delta_k}})_i}$. 
 \item If $\theta(\wopt^{\delta_k}) > c+\epsilon$, set $\delta_{k+1}=2\delta_k$,  set $k \leftarrow k+1.$ and go to step 1.
 \item If $\theta(\wopt^{\delta_k}) < c-\epsilon$, do a binary search with $\delta$ lying between lower limit as $\delta_{\text{min}}=\delta_{k-1}$ and upper limit as $\delta_{\text{max}}=\delta_{k}$ till we obtain a $\delta$ such that $|\theta(\wopt^\delta)-c|\leq \epsilon.$
  \end{enumerate}
}

In practice, we are interested in obtaining $s$-values over an arbitrary range of constants $c$, in which case, we can just fix a range of values for the penalty $\delta$ in increasing order (say $\delta_0 <\delta_1< \ldots < \delta_P$ for some $P \in \Z_{+}$) and use Algorithm~\ref{alg:r-M} to obtain corresponding $s$-value for a given $\delta \in \{\delta_1, \ldots \delta_P\}$ where we can now use warm start to initialize the algorithm for $\delta_p$ using the final iterate of the algorithm for $\delta_{p-1}$. Such heuristics give efficiency gain in practice.
%We present the proof of Proposition~\ref{prop:conv-optim} in Appendix~\ref{proof:conv-optim}.

% Hence, we always obtain a locally optimal solution to the problem in \eqref{eqn:r-value-finite} using the above gradient descent-based method.

\begin{remark}
The above procedure generalizes to the directional case \eqref{eqn:r-condn}. However, it requires obtaining the conditional expectation of the gradient of the parameter $\theta$ with respect to the variable $E$. We can get an exact estimate of the conditional expectation when $E$ has finite support and similar analysis as above guarantees convergence of the iterates to local optima. However, if $E$ has infinite support (for example, $E$ is a continuous random variable) then we can only obtain an approximation of the conditional expectation using (say) any non-parametric regression method in which case we do not have a guaranteed convergence to local optima. In such situations, we can modify the problem by discretizing $E$ to have such guarantees. We give more details in the next section \ref{sec:alg-cond-s}.
\end{remark}
\subsection{Proof of Proposition~\ref{prop:conv-optim}}
\label{proof:conv-optim}

We next proceed to prove Proposition~\ref{prop:conv-optim}. The proof uses similar arguments as the proof of Proposition 12.4.4 of \citet{Lange13}. The proof builds on the following lemmas.

\begin{definition}[Cluster point of a sequence]
A point $w^*$ is a cluster point of a sequence $w^k$ provided there is a subsequence $w^{k_l}$ that tends to $w^*$.
\end{definition}
\begin{lemma}[Proposition 12.4.1, \citet{Lange13}]
\label{lem:12.4.1}
If a bounded sequence $w^k \in \R^n$ satisfies
\begin{align*}
\lim_{k \to \infty} \norm{w^{k+1}-w^k}=0,
\end{align*}
then its set $T$ of cluster points is connected. If $T$ is finite, then $T$ reduces to a single point, and $\lim_{k \to \infty} w^k =w^*$ exists.
\end{lemma}

\begin{lemma}
Let $\Gamma$ be the set of cluster points generated by the MM sequence $w^{k+1}=M(w^k)$ starting from some initial $w^0$. then $\Gamma$ is contained in the set $S$ of stationary points of \eqref{eqn:general-optim-lag}.
\end{lemma}
\begin{proof}
First observe that the iteration map $M$ in \eqref{eqn:next-iter} is continuous as $\theta$ is continuously differentiable. Now, the sequence $w^k$ stays within the compact set $W_n$. Consider a cluster point $z= \lim_{l \to \infty} w^{k_l}$. Since the sequence $g(w^k)$ is monotonically decreasing and bounded below, $\lim_{k \to \infty} g(w^k)$ exists. Hence, taking limits in the inequality $g(M(w^k)) \leq g(w^k)$ and using the continuity of functions $M$ and $g$ imply $g(M(z))=g(z)$. Thus, $z$ is a fixed point of $M$ and also a stationary point of \eqref{eqn:general-optim-lag}.
\end{proof}

\begin{lemma}
The set of cluster points $\Gamma$ of $w^{k+1}=M(w^k)$ is compact and connected.
\end{lemma}

\begin{proof}
$\Gamma$ is a closed subset of the compact set $W_n$ and is hence, compact. By Lemma \ref{lem:12.4.1}, $\Gamma$ is connected provided $\lim_{k \to \infty} \norm{w^{k+1}-w^k}=0.$ If this sufficient condition fails, then by compactness of $W_n$, we can extract a subsequence $w^{k_l}$ such that $\lim_{l \to \infty} w^{k_l}=u$ and $\lim_{l \to \infty}w^{k_l +1}=v$ both exist, however, $v \neq u$. Further, continuity of function $M$ requires $v=M(u)$ while the descent condition implies 
\begin{align*}
g(v)=g(M(u))=g(u)=\lim_{k \to \infty} g(w^k).
\end{align*}
Hence, $u$ is a fixed point of $M$, which is a contradiction. Hence, the sufficient condition that $\lim_{k \to \infty} \norm{w^{k+1}-w^k}=0$ holds.
\end{proof}

From \eqref{eqn:upper-bound-alg}, we observe that $G_L(w',w)$ is strictly convex in $w'$ and hence, we have the following chain of inequalities
\begin{align}
g(w^{k+1}) \leq G_L(w^{k+1},w^k) <G_L(w^k,w^k) =g(w^k).
\end{align}
Since $g$ is lower bounded, hence the sequence $g(w^k)$ decreases and converges which proves 1. 

Now, if all stationary points of \eqref{eqn:general-optim-lag} are isolated and since the domain $W_n$ is compact, then there can only be a finite number of stationary points as an infinite number of them would admit a convergent sequence whose limit will not be isolated. Since, the set of cluster points $\Gamma$ of $w^{k+1}=M(w^k)$ is a connected subset of the finite set of stationary points, $\Gamma$ is a singleton, and hence, the bounded sequence $w^k$ has the single element of $\Gamma$ as its limit. Let $\lim_{k \to \infty} w^k= w^*$, then by Proposition \ref{prop:solve-upper-bound}, we have $w^*_i \propto e^{-\delta \nabla_i \theta(w^k)}$ for all $i=\{1,2, \ldots, n\}$. Hence, by Corollary~\ref{thm:optim-result-finite}, we have the result.

\section{Directional $s$-values of general estimands}
\label{sec:alg-cond-s}

Here, we want to obtain directional $s$-values (with respect to some variable $E$) as in \eqref{eqn:r-condn} for more general one dimensional parameters defined over the space of probability distributions, $\theta : \mc{P} \to \R$. We first characterize the form of a locally optimal solution of the optimization problem in \eqref{eqn:r-condn} and present algorithm to solve the corresponding optimization problem in Appendix \ref{sec:alg-cond-s-alg}.

We assume that random variable $E$ has finite support of size $K$ (say) and $E$ takes values in the set $\{e_1, \ldots, e_K\}$. We consider a finite sample setting where we observe $n$ samples $\{Z_i, E_i\}_{i=1}^n \overset{\text{i.i.d.}}{\sim} P_0$ for some distribution $P_0 \in \mc{P}$ where $\{Z_i,E_i\}_{i=1}^n$ are i.i.d.\ realizations of the random variable $(Z,E)$.  Let the empirical distribution of $\{Z_i, E_i\}_{i=1}^n $ be denoted by $P_{0,n}=\sum_{i=1}^n \frac{1}{n} \delta_i$, where $\delta_i$ is a dirac measure on $(Z_i, E_i)$. We recall that $W_n=[0,1]^n$ is $n$ dimensional unit cube and $S_n=\{w \in \R^n : w_0+ \ldots + w_n=1, w_i \geq 0 \text{ for } i=1, \ldots, n\}$ is $n$ dimensional probability simplex. Let $P_w$ denote the probability distribution corresponding to $w \in S_n$ that is, it puts mass $w_i$ on the $i$th sample. We focus on one dimensional parameter $\theta : S_n \to \R$ where we define for $w \in S_n$, $\theta(w) $ as $\theta(\sum_{i=1}^n w_i \delta_i)$.
With a slight abuse of notation from now on, we redefine $\theta$ on the $n$ dimensional unit cube $W_n$ as $\theta(w)=\theta \left( \frac{\sum_{i=1}^{n} w_{i} \delta_{i}}{\sum_i w_i} \right)  $ for $w \in W_n$.
We recall that we want to obtain the conditional $s$-value of parameter $\theta$ (with respect to the variable $E$) given by

 %Let $C=\{p \in \R^n : p_0+ \ldots + p_n=1, p_i \geq 0 \text{ for } i=1, \ldots, n\}$ be $n$ dimensional probability simplex.
\begin{align}
\label{eqn:r-value-finite-cond}
\begin{split}
   s_E(\theta-c,P_{0,n})=\exp\{-\min_{w \in W} \sum_{i=1}^n w_i \log(n w_i)\} &\hspace{0.1in} \text{ s.t. } \theta(w) = c, \text{ }\sum_{i=1}^n w_i=1 \text{ and }\\ &P_{0,n}(\cdot \mid E=e_k)= P_w(\cdot \mid E=e_k) \text{ for all } k\in [K].
   \end{split}
\end{align}
The constraints $P_{0,n}(\cdot \mid E=e_k)= P_w(\cdot \mid E=e_k) \text{ for all } k\in [K]$ are linear in weights $w$ that we justify next. Let $I_k$ denote the set of indices such that $E_j=e_k$ for all $j \in I_k$ and each $k \in [K]$, then we have for each $k \in [K]$ and $i \in I_k$
\begin{align}
\begin{split}
P_{0,n}(Z_i \mid E=e_k)&= P_w(Z_i \mid E=e_k)  \\
\implies \frac{w_i}{\sum_{j \in I_k} w_j}=\frac{1}{|I_k|}.
\end{split}
\end{align}
Hence, the above constraint implies that for each $k \in [K]$, all $w_i$ such that $i \in I_k$ are equal. That is, the constraints $P_{0,n}(\cdot \mid E=e_k)= P_w(\cdot \mid E=e_k) \text{ for all } k\in [K]$ are equivalent to the  constraint that $w_i = w_j$ for all $(i,j)$ with $E_i = E_j$. We can rewrite the above constraints by a collection of pairwise equality constraints using a minimum collection of functions $\mc{U}$ such that for any $u: W_n \to \R$  such that $u \in \mc{U}$, $u$ is given by $u(w)=w_a-w_b$ for some $a\neq b $ where $a,b \in [n]$. Hence, the above optimization problem belongs to the class of general constrained minimization problems with equality constraints (see Chapter 3 of \citet{Bertsekas99}). Now we present necessary and sufficient conditions for a point to be a local optimum of \eqref{eqn:r-value-finite-cond}, which can be used to verify that we obtained a locally optimal solution of our optimization problem \eqref{eqn:r-value-finite-cond}.
Let $M$ be a random variable taking values in the set $\{\nabla_1 \theta(w), \ldots, \nabla_n \theta(w)\}$. Now, for any given probability distribution $P\in \mc{P}$, let there be a probability distribution $Q$ such that $\{Z_i, E_i,M_i\}_{i=1}^n \overset{\text{i.i.d.}}{\sim}Q$ where $Q$ is the push-forward of $(Z,E) \sim P$, that is, $Q((Z,E,M)=(Z_i,E_i,\nabla_i \theta))=P((Z,E)=(Z_i,E_i))$ for $i \in [n]$ and $P \in \mc{P}$. In particular, we denote the push forward of $(Z,E) \sim P_0$ under the above mapping by $Q_0$.

We first give a necessary condition for a point to be a local optimum of \eqref{eqn:r-value-finite-cond} that follows from Proposition 3.1.1 of \citet{Bertsekas99}. 

\begin{corollary}[Necessary conditions]
\label{thm:optim-result-finite-cond}
Assume that $\theta : \mathrm{int}(W_n) \to \R$ is continuously differentiable. Let $w^* \in S_n$ be a locally optimal solution to problem \eqref{eqn:r-value-finite-cond}, and assume that there does not exist a constant $r \in \R$ such that $(\E_{Q_0}[M \mid E=e_1], \ldots, \E_{Q_0}[M \mid E=e_K])=r (1, \ldots, 1)$. Then there exists a constant $\lambda \in \R$ such that
\begin{equation}
w^*_i \propto e^{\lambda \E_{Q_0}[M \mid E=E_i]} \text{ for all } i=\{1, \ldots, n\}.
\end{equation}
\end{corollary}

\begin{proof}
Under the given assumption of Corollary~\ref{thm:optim-result-finite-cond}, the assumption that vectors $\nabla \theta$, $(1, \ldots, 1)$, $\nabla_w u$ for $u \in \mc{U}$ are linearly independent holds as otherwise we get a contradiction.
Now, without loss of generality, we assume that $Z_i$'s are distinct and let $E_1=E_2=\ldots=E_m=e_1$ for some $m <n$. We show that 
\begin{align*}
w_1=w_2=\ldots=w_m \propto e^{\lambda \E_{P_0}[M \mid E=e_1]} .
\end{align*}
We need to take the derivative of the Lagrangian \eqref{eqn:lag-cond}. Without loss of generality, let the functions in $\mc{U}$ corresponding to the pair wise equality of $w_1,w_2, \ldots, w_m$ be given by 
\begin{align*}
u_1(w)&=w_1-w_2 \\
u_2(w)&=w_1-w_3 \\
u_3(w)&=w_1-w_3 \\
&\vdots\\
u_{m-1}(w)&=w_1-w_m. \\
\end{align*}
Other functions $u \in \mc{U}$ do not depend on any of $w_1, \ldots, w_m$.

Hence, the Lagrangian now becomes

\begin{align}
\begin{split}
h(w,\delta, \mu)= \sum_{i=1}^{n} w_{i} \log \left( w_{i} \right) +\delta (\theta \left( w\right)-c) &+\mu(\sum_{i=1}^n  w_i-1) +\sum_{k=1}^{m-1}\alpha_i (w_1-w_{i+1})+ \sum_{u \in \mc{U}-\{u_1, \ldots, u_{m-1}\}} \alpha_u u\\
& \text{ for } w \in W_n \text{, and } \delta,\mu, \alpha_u \in \R.
\end{split}
\end{align}
Taking partial derivatives of $h$ with respect to $w_1, \ldots, w_m$, we get

\begin{align*}
\log w_1+1+\delta \nabla_1\theta(w) +\mu+\alpha_1+\ldots+\alpha_{m-1}&=0\\
\log w_2+1+\delta \nabla_2\theta(w) +\mu-\alpha_1&=0\\
&\vdots\\
\log w_m+\delta \nabla_m\theta(w) +\mu-\alpha_{m-1}&=0.
\end{align*}
Now, invoking the constraint $w_1=\ldots=w_m$, and adding the above equations, the result follows from Proposition 3.1.1 of \citet{Bertsekas99}.
\end{proof}

We next present the sufficient condition for a point to be local optima of \eqref{eqn:r-value-finite-cond} that again follows from Proposition 3.2.1 of \citet{Bertsekas99}. To that end, we introduce the Lagrangian function $h : \R^n\times \R \times \R \to \R$ that we define as
\begin{equation}
\label{eqn:lag-cond}
h(w,\delta, \mu)= \sum_{i=1}^{n} w_{i} \log \left( w_{i} \right) +\delta (\theta \left( w\right)-c) +\mu(\sum_{i=1}^n  w_i-1) +\sum_{u \in \mc{U}} \alpha_u u \text{ for } w \in W_n \text{, and } \delta,\mu, \alpha_u \in \R.
\end{equation}
\begin{corollary}[Second order sufficiency conditions]
Assume that $\theta  : \mathrm{int}(W_n) \to \R$ is twice continuously differentiable, and let $w^* \in W_n$, $\delta^*,\mu^* \in \R$ and $\alpha^* \in \R^{|U|}$ satisfy
\begin{align*}
&\nabla_w h(w^*,\delta^*,\mu^*, \alpha^*)=0, \text{ } \nabla_{\delta,\mu, \alpha} h(w^*,\delta^*,\mu^*, \alpha^*)=0, \\
& \gamma' \nabla_{ww}^2 h(w^*,\delta^*, \mu^*, \alpha^*)\gamma>0, \text{ for all }\gamma \neq 0 \text{ with } \\
&\nabla \theta(w^*)'\gamma=0, \sum_{i=1}^n \gamma_i=0   \text{ and } \nabla u(w^*)'\gamma=0 \text{ for all } u \in \mc{U}.
\end{align*}
Then $w^*$ is a strict local optima of \eqref{eqn:r-value-finite-cond}.
\end{corollary}
Next, we present a Majorization-minimization based algorithm to solve \eqref{eqn:r-value-finite-cond} that relies on this characterization.

\subsection{Algorithms to obtain directional $s$-values of general estimands}
\label{sec:alg-cond-s-alg}
Here, we solve the optimization problem in \eqref{eqn:r-value-finite-cond}.
Following similar arguments as in Section~\ref{sec:alg}, we solve the Lagrangian form given by

\begin{equation}
\label{eqn:general-optim-lag-cond}
  \minimize_{w_{1},\ldots,w_{n}, w_{i} \ge 0, \sum w_{i} =1, P_{0,n}(\cdot \mid E)=P_w(\cdot \mid E)}  g(w)=\delta (\theta \left( w\right)-c) + \sum_{i=1}^{n} w_{i} \log \left( w_{i} \right).
\end{equation}

We solve \eqref{eqn:general-optim-lag-cond} using Majorization-Minimization algorithm. We obtain the majorizer of the objective function in \eqref{eqn:general-optim-lag-cond} using \eqref{eqn:upper-bound} under Assumption~\ref{ass:continuous-differentiability} as follows

\begin{align}
\label{eqn:upper-bound-alg-cond}
g(w') \leq G_{L}(w',w) \defeq \delta \left(\theta(w)-c + \langle \nabla \theta(w), w'-w \rangle + L\sum_{i=1}^n w_i'\log \frac{w_i'}{w_i}\right) +\sum_{i=1}^n w_i'\log w_i'
\end{align}
for $w,w' \in W_n$.

First we observe that $\langle \nabla \theta(w), w'-w \rangle= \E_{Q_{w'}}[M]-\E_{Q_w}[M]$.  Now we want to minimize the right hand side of inequality \eqref{eqn:general-optim-lag-cond} with respect to $w'$ under the additional constraint $P_{0,n}(\cdot \mid E=e_k)=P_{w'}(\cdot \mid E=e_k)$ for all $k\in [K]$ which gives 
\begin{align*}
\langle \nabla \theta(w), w'-w \rangle 
&= \E_{Q_{w'}}[M]-\E_{Q_w}[M]  \\
&= \E_{Q_{w'}}[\E_{Q_{w'}}[M \mid E]]-\E_{Q_w}[M] \\
&=\E_{Q_{w'}}[\E_{Q_{0,n}}[M \mid E]]-\E_{Q_w}[M] .
\end{align*}
Hence, under the additional constraint, the majorizer now becomes

\begin{align}
\label{eqn:upper-bound-alg-cond-cons}
\begin{split}
g(w') \leq   G_{L}(w',w) \defeq \delta  & \left(\theta(w)-c + \E_{Q_{w'}}[\E_{Q_{0,n}}[M \mid E]]-\E_{Q_w}[M] 
+ L\sum_{i=1}^n w_i'\log \frac{w_i'}{w_i}\right) \\
& +\sum_{i=1}^n w_i'\log w_i'
\end{split}
\end{align}
for $w,w' \in W_n$.

We next show that minimizing the majorizer $G_L(w',w)$ actually involves solving a $K$-dimensional convex optimization problem. The random variable $E$ takes values in the set $\{e_1, \ldots, e_K\}$. Suppose out of the $n$ realizations $\{Z_i, E_i\}_{i=1}^n$, $e_k$ occurs $n_k$ times for $k \in [K]$ and $\sum_{k=1}^Kn_k=n$. Now under the constraint $P_{0,n}(\cdot \mid E)=P_{w'}(\cdot \mid E)$, it is equivalent to considering only probability distributions on the set $\{e_1, \ldots, e_K\}$ as conditional on $E=e_k$ for any $k \in [K]$, the corresponding samples are equally likely to occur. Hence, now we can restrict our domain to $K$ dimensional unit cube $W_K$ and minimizing the majorizer in \eqref{eqn:upper-bound-alg-cond-cons} is equivalent to solving the following optimization problem

\begin{equation}
\minimize_{v' \in W_K, \sum_{k=1}^Kv_k^{'}=1} \delta \left(\sum_{k=1}^K v_k^{'} \E_{Q_{0,n}}[M \mid E=e_k] +L \sum_{k=1}^K v_k^{'}\log \frac{v_k^{'}}{v_k}\right)+\sum_{k=1}^K v_k^{'}\log \frac{v_k^{'}}{n_k}\
\end{equation}
which is a $K$ dimensional convex optimization problem. Hence, the convergence analysis follows as in Section~\ref{sec:alg}.

If the variable $E$ is continuous-valued, then we can discretize $E$ to use the similar procedure as outlined above or use any non-parametric estimator to approximate the conditional expectation $\E_{Q_0}[M \mid E]$.

\section{Confidence intervals}\label{sec:proof-ci}

In this section, we prove a general theorem that contains the asymptotic normality results from Section~\ref{sec:cons-norm} as a special case. In particular, Lemma~\ref{lem:asy-normality-mean} can be recovered with $E=Z = \ell(\eta,Z)$ and $\hat f_n(E) = E$, $f(E) = E$. Lemma~\ref{lem:asy-normality-cmean} can be recovered with $Z = \ell(\eta,Z)$.
\begin{theorem}
  Let $\hat f_n(\cdot) $ be an estimate of $f(\cdot) = \mathbb{E}_{P_0}[\ell(\eta,Z)|E=\cdot]$. We assume that $\hat f_n$ and $\hat \lambda$ are fit on a held-out portion of the data set, that is $\hat f_n(\cdot)$ and $\hat \lambda$ are independent of $(Z_i,E_i), i=1,\ldots,n$. We assume that $ \sup_{e \in \mathcal{E}} | \hat f_n(e) - f(e)  | = o_P(n^{-1/4})$. Furthermore, we assume that the moment generating function of $\ell(\eta,Z_i)$ is finite on $\mathbb{R}^p$ and that $\mathbb{E}[f(E) f(E)^\intercal e^{(\lambda^*)^\intercal  f(E)}] > 0$.  Let $\hat \lambda = \arg \min \frac{1}{n} \sum_{i=1}^n e^{\lambda^\intercal \hat f_n(E_i)}$ and $\lambda^* = \arg \min \mathbb{E}_{P_0}[e^{(\lambda)^\intercal f(E)}]$. Then,
  \begin{equation*}
   \frac{1}{n} \sum_{i=1}^n (1 + \hat \lambda^\intercal \ell(\eta,Z_i) - \hat \lambda^\intercal \hat f_n(E_i)) e^{\hat \lambda^\intercal \hat f_n(E_i)}    - \mathbb{E}_{P_0}[e^{(\lambda^*)^\intercal f(E)}] \stackrel{d}{=} \mathcal{N} \left(0,\frac{\sigma^2}{n} \right) + o_P(1/n),
  \end{equation*}
  where
  \begin{equation*}
     \sigma^2 =  \text{Var}_{P_0}(e^{\lambda^* f(E)})  +  \text{Var}_{P_0}( e^{(\lambda^*)^\intercal f(E)} (\lambda^*)^\intercal ( \ell(\eta,Z) - f(E))).
  \end{equation*}
  \end{theorem}
  
  \begin{proof}
  Using Lemma~\ref{lem:consistency-mean-cond}, we have $\hat \lambda \rightarrow \lambda^*$ in probability. By definition of $\hat \lambda$,
  \begin{equation*}
    \frac{1}{n} \sum_{i=1}^n  \hat f_n(E_i) e^{\hat \lambda^\intercal \hat f_n(E_i)} = 0.
  \end{equation*}
  %We can now add the same term on both sides:
  % \begin{equation*}
  %   \frac{1}{n} \sum_{i=1}^n  \hat f_n(E_i) e^{\hat \lambda^\intercal \hat f_n(E_i)} - \frac{1}{n} \sum_{i=1}^n f(E_i) e^{(\lambda^*)^\intercal f(E_i)}   = - \frac{1}{n} \sum_{i=1}^n f(E_i) e^{(\lambda^*)^\intercal f(E_i)}
  % \end{equation*}
  Using a Taylor expansion on the left, %and the CLT on the right,
  \begin{equation*}
      \frac{1}{n} \sum_{i=1}^n f(E_i) e^{\hat \lambda^\intercal  f(E_i)} = o_P(n^{-1/4}).
  \end{equation*}
  Thus, using another Taylor expansion,
  \begin{equation*}
      (\hat \lambda - \lambda^*)^\intercal \frac{1}{n}\sum_{i=1}^n f(E_i) f(E_i)^\intercal e^{(\lambda^*)^\intercal  f(E_i)} = 
      O_P(\|\hat \lambda - \lambda^*\|_2^2 + \frac{1}{\sqrt{n}})  + o_P(n^{-1/4}).
  \end{equation*}
  Since $\hat \lambda - \lambda^* \rightarrow 0$ and $\mathbb{E}[ f(E) f(E)^\intercal e^{ (\lambda^*)^\intercal  f(E)}] > 0$ we have
  \begin{equation*}
      \hat \lambda - \lambda^* = o_P(n^{-1/4}).
  \end{equation*}
  To summarize, we know that $\hat \lambda - \lambda^* = o_P(n^{-1/4})$ and that $ \sup_{e \in \mathcal{E}} | \hat f_n(e) - f(e) | = o_P(n^{-1/4})$. %We consider $Q$ and $\lambda$ to be independent of $D_1,\ldots,D_n$.
  Thus,
  \begin{align*}
        \frac{1}{n} \sum_{i=1}^n e^{\hat \lambda^\intercal \hat f_n(E_i)} - \mathbb{E}[e^{(\lambda^*)^\intercal f(E)}] &=  (\hat \lambda - \lambda^0)^\intercal \frac{1}{n} \sum_{i=1}^n f(E_i) e^{(\lambda^*)^\intercal f(E_i)} \\
    &+ \frac{1}{n} \sum_{i=1}^n (\hat \lambda)^\intercal (\hat f_n(E_i)-f(E_i)) e^{( \hat \lambda)^\intercal \hat f_n(E_i)} \\
    & + \frac{1}{n} \sum_{i=1}^n e^{(\lambda^*)^\intercal f(E_i)} - \mathbb{E}_{P_0}[e^{(\lambda^*)^\intercal f(E)}] + o_P(n^{-1/2}) \\
  \end{align*}
  Using the CLT, the first term goes to zero at rate $O_{P}(n^{-3/4})$. The last term can be computed with a CLT (since we assume that the moment generating function is finite, the variance is finite). Let us focus on the second term. 
  Note that the second term in the previous equation can be re-written:
  \begin{align*}
     \frac{1}{n} \sum_{i=1}^n (\hat \lambda)^\intercal  ( \hat f_n(E_i) -  f(E_i)) e^{(\hat \lambda)^\intercal \hat f_n(E_i) } &= \frac{1}{n} \sum_{i=1}^n (\hat \lambda)^\intercal  ( \ell(\eta,Z_i) -  f(E_i)) e^{(\hat \lambda)^\intercal \hat f_n(E_i) }  \\
    &+ \frac{1}{n} \sum_{i=1}^n (\hat \lambda)^\intercal ( \hat f_n(E_i) -  \ell(\eta,Z_i)) e^{(\hat \lambda)^\intercal \hat f_n(E_i) } 
  \end{align*}
Using the last two equations,
\begin{align}\label{eq:big-form}
  \begin{split}
  &\frac{1}{n} \sum_{i=1}^n (1 + \hat \lambda^\intercal \ell(\eta,Z_i) - \hat \lambda^\intercal \hat f_n(E_i)) e^{\hat \lambda^\intercal \hat f_n(E_i)}    - \mathbb{E}_{P_0}[e^{(\lambda^*)^\intercal f(E)}]\\
  &= \frac{1}{n} \sum_{i=1}^n (\hat \lambda)^\intercal  ( \ell(\eta,Z_i) -  f(E_i)) e^{(\hat \lambda)^\intercal \hat f_n(E_i) } + \frac{1}{n} \sum_{i=1}^n e^{(\lambda^*)^\intercal f(E_i)} - \mathbb{E}_{P_0}[e^{(\lambda^*)^\intercal f(E)}] + o_P(n^{-1/2})
  \end{split}
\end{align}
Using that 
  \begin{equation*}
  \mathbb{E}_{P_0}[  ( \ell(\eta,Z) -  f(E)) | E ] = 0,
  \end{equation*}
  by conditioning on the $E_1,\ldots,E_n$ using a CLT we get
  \begin{equation*}
      \frac{1}{n} \sum_{i=1}^n   ( \ell(\eta,Z_i) -  f(E_i)) e^{(\hat \lambda)^\intercal \hat f_n(E_i) } = \frac{1}{n} \sum_{i=1}^n   ( \ell(\eta,Z_i) -  f(E_i)) e^{(\lambda^*)^\intercal f(E_i) } + o_P(1/\sqrt{n}).
  \end{equation*}
  Here we used that $\hat f_n$ is computed on a separate data set and thus is independent of $(Z_i,E_i)$, $i=1,\ldots,n$. Furthermore, we used that $\hat \lambda$ depends on the $(Z_i,E_i)$ only through the $E_i$. Using this in equation~\eqref{eq:big-form}, we get
  \begin{align*}
      &\frac{1}{n} \sum_{i=1}^n (1 + \hat \lambda^\intercal \ell(\eta,Z_i) - \hat \lambda^\intercal \hat f_n(E_i)) e^{\hat \lambda^\intercal \hat f_n(E_i)}    - \mathbb{E}_{P_0}[e^{(\lambda^*)^\intercal f(E)}] \\
      %&=  \frac{1}{n} \sum_{i=1}^n ( \hat \lambda)^\intercal  ( \ell(\eta,Z_i) -  f(E_i)) e^{(\hat \lambda)^\intercal \hat f_n(E_i) } \\
      %&+ \frac{1}{n} \sum_{i=1}^n e^{(\lambda^*)^\intercal f(E_i)} - \mathbb{E}_{P_0}[e^{(\lambda^*)^\intercal f(E)}] + o_P(1/\sqrt{n}) \\
      &= \frac{1}{n} \sum_{i=1}^n (  \lambda^*)^\intercal  ( \ell(\eta,Z_i) -  f(E_i)) e^{( \lambda^*)^\intercal  f(E_i) } \\
      &+ \frac{1}{n} \sum_{i=1}^n e^{(\lambda^*)^\intercal f(E_i)} - \mathbb{E}_{P_0}[e^{(\lambda^*)^\intercal f(E)}] + o_P(1/\sqrt{n})
  \end{align*}
  Thus, asymptotically, we have that
  \begin{equation*}
  \frac{1}{n} \sum_{i=1}^n (1 + \hat \lambda^\intercal \ell(\eta,Z_i) - \hat \lambda^\intercal \hat f_n(E_i)) e^{\hat \lambda^\intercal \hat f_n(E_i)}    - \mathbb{E}_{P_0}[e^{(\lambda^*)^\intercal f(E)}]  \stackrel{d}{=} \mathcal{N}(0,\sigma^2) + o_P(1/\sqrt{n}),
  \end{equation*}
  where
  \begin{equation*}
     \sigma^2 = \frac{1}{n} \text{Var}(e^{(\lambda^*)^\intercal f(E_i)})  + \frac{1}{n} \text{Var}( e^{(\lambda^*)^\intercal f(E_i)} (\lambda^*)^\intercal ( \ell(\eta,Z) - f(E))).
  \end{equation*}
  \end{proof}

\section{Other technical proofs and appendices}

%\begin{theorem}
%  \label{thm:Donsker-Varadhan-argmin}
% For a probability distribution $P_0$, let $Z \sim P_{0}$ be a real-valued random variable taking values in $\mc{Z} \subseteq \R$ with mean $\mu(P_0)=\mathbb{E}_{P_0}[Z]$ and finite moment generating function. Let $Q$ be a probability distribution defined as
% \begin{align}
% \label{eqn:DV-argmin}
%Q=  \argmin_{P \in \mc{P}} D_{KL}(P||P_0)\}  \text{ s.t. } \	E_{P_0}[Z] = 0.
% \end{align}
% If $ \inf_{\lambda}  \mathbb{E}_{P_0}[e^{\lambda Z}]$ is attained at some $\lambda^* \in \R$ then $Q$ as defined in \eqref{eqn:DV-argmin} is given by
% \begin{align*}
% dQ(z)=\frac{e^{\lambda^* z}}{\E[e^{\lambda^* Z}]} dP_0(z) \text{ for all } z \in \mc{Z}.
% \end{align*}
%  
%\end{theorem}
%\begin{proof}
%The proof of the theorem follows from the proof of Theorem 5.2 of \citet{DonskerVa76}.
%\end{proof}

\begin{theorem}\label{thm:uniform-convex}[Theorem II.1, \citet{AndersenGi82}]
Let $E$ be an open convex subset of $\R^p$ and let $F_1, F_2, \ldots,$ be a sequence of random concave functions on $E$ such that $F_n(x) \overset{P}{\to} f(x)$ as $n \to \infty$ for every $x \in E$, where $f$ is some real function on $E$. Then $f$ is also concave and for all compact $A \subset E$,
\begin{equation*}
\sup_{x \in A} |F_{n}(x)-f(x)| \overset{P}{\to} 0 \text{ as } n \to \infty.
\end{equation*} 
\end{theorem}

\begin{corollary}\label{cor:argmin-conv}[Corollary II.1, \citet{AndersenGi82}]
Let $E$ be an open convex subset of $\R^p$ and let $F_1, F_2, \ldots,$ be a sequence of random concave functions on $E$ such that $F_n(x) \overset{P}{\to} f(x)$ as $n \to \infty$ for every $x \in E$, where $f$ is some real function on $E$. Suppose $f$ has a unique maximum at $\hat x \in E$. Let $\hat X_n$ maximize $F_n$. Then $\hat X_n \overset{P}{\to} \hat x$ as $n \to \infty$.
\end{corollary}
\begin{corollary}
\label{cor:min-conv}
Let $E$ be an open convex subset of $\R^p$ and let $F_1, F_2, \ldots,$ be a sequence of random concave functions on $E$ such that $F_n(x) \overset{P}{\to} f(x)$ as $n \to \infty$ for every $x \in E$, where $f$ is some real function on $E$. Suppose $f$ has a unique maximum at $\hat x \in E$. Let $\hat X_n$ maximize $F_n$. Then $F_n(\hat X_n) \overset{P}{\to} f(\hat x)$ as $n \to \infty$.
\end{corollary}
\begin{proof}
 We define a set $B$ as $B=\{x : \norm{x-\hat x} \leq \gamma\}$ for some arbitrary small $\gamma>0$ such that $B \subseteq E$. Clearly, set $B$ is compact. From Corollary~\ref{cor:argmin-conv}, we have $\hat X_n \overset{P}{\to} \hat x.$ Hence, there exists positive integer $N_1$ such that $\hat X_n \in B$ for all $n>N_1$ with probability at least $1-\delta$ for some small $\delta>0$.

Since $\sup_{x \in B}|F_n(x) -f(x)| \overset{P}{\to} 0$. Hence, for any $\epsilon >0$, there exists positive integer $N_2$ such that 
\begin{equation}
\label{eqn:uniform-conv}
|F_n(x) -f(x)| <\epsilon \text{ for all } x \in B \text{ and } n>N_2
\end{equation}
with probability at least $1-\delta$.
Let $x_0 \in B$ be such that $f(x_0) \geq \sup_{x \in B} f(x) -\epsilon$. Hence, using \eqref{eqn:uniform-conv}, we have for all $n>N_2$
\begin{align}
\label{eqn:lower-conv}
\sup_{x \in B} f(x) &\leq  f(x_0) +\epsilon  \leq F_n (x_0) +2\epsilon \leq \sup_{x \in B}F_n(x)+2\epsilon
\end{align}
with probability at least $1-\delta$. 

Now, we choose sequence $x_n \in B$ such that $F_n(x_n) \geq \sup_{x \in B} F_n(x) -\epsilon$. Using \eqref{eqn:uniform-conv}, we have for all $n>N_2$
\begin{align}
\label{eqn:upper-conv}
\sup_{x \in B} f(x) +\epsilon \geq  F_{n}(x_n) \geq \sup_{x \in B} F_{n}(x) -\epsilon
\end{align}
with probability at least $1-\delta$.
Combining \eqref{eqn:lower-conv} and \eqref{eqn:upper-conv}, we have for all $n>N_2$,
\begin{align}
|\sup_{x \in B} F_n(x) - \sup_{x \in B} f(x)| < 2\epsilon
\end{align}
with probability at least $1-\delta$. 
We choose $N=\max\{N_1,N_2\}$. Since, $\hat X_n \in B$ for all $n>N$ with probability at least $1-\delta$. We have for all $n>N$, with probability at least $1-2\delta$,

\begin{align}
|F_n(\hat X_n) -  f(\hat x)| < 2\epsilon.
\end{align}

Hence, the proof follows.
\end{proof}
\subsection{Proof of Lemma~\ref{lem:asy-norm-mean}}
\label{proof:asy-norm-mean}

First, if $Z \ge 0$ with probability 1 or if $Z \le 0$ with probability 1, the statement is immediate. Thus, in the following we assume that $Z > 0 $ with non-vanishing probability and that $Z < 0 $ with non-vanishing probability.

Let us now prove that the minimum is achieved for some unique $\lambda^* \in \mathbb{R} \cup \{ - \infty, \infty \} $. We will do the proof by contradiction. If the minimum is attained for multiple $\lambda^*$, then by convexity there must exist a nonempty open interval $(\lambda_1,\lambda_2)$ of values $\lambda^*$ that attain the minimum. Using a second order Taylor expansion, one can show that in this case we must have $Z \equiv c$ almost surely. However, we assumed that $Z$ is non-degenerate. Thus, the minimum is achieved for some unique $\lambda^* \in \mathbb{R} \cup \{ - \infty, \infty \} $.

Furthermore, if $Z > 0$ with probability $>0$ then $\mathbb{E}[e^{\lambda Z}] \rightarrow \infty$ for $\lambda \rightarrow \infty$. Similarly if $Z < 0 $ with probability $< 0$ then $\mathbb{E}[e^{\lambda Z}] \rightarrow \infty$ for $\lambda \rightarrow -\infty$. Thus, the minimum is achieved for $\lambda^* \in \mathbb{R}$.

The proof follows from Corollary \ref{cor:min-conv} and using the fact that the negative of a convex function is concave.

\subsection{Proof of Theorem~\ref{thm:Donsker-Varadhan-condn}}\label{sec:proof-cond}

\begin{proof}
Any distribution $\mathbb{P}$ that satisfies $\mathbb{P}[\cdot |E=e] = \mathbb{P}_{0}[ \cdot|E=e] $ for all $e \in \mc{E}$ satisfies
\begin{equation}
  \mathbb{E}_P[Z] = \mathbb{E}_P[\mathbb{E}_{P_{0}}[Z|E]].
\end{equation}
Thus,
\begin{align*}
  s_E(\theta,P_0) &=\exp\{-\min_{P \in \mc{P} :  P(\cdot | E=e ) = P_0(\cdot | E=e)  \text{ for all} e \in \mc{E}} D_{KL}(P||P_{0})\} \hspace{0.1in} \text{s.t.} \hspace{0.1in}  \mathbb{E}_P[Z] = 0. \\
  &= \exp\{-\min_{P \in \mc{P} :  P(\cdot | E=e) = P_0(\cdot | E=e) \text{ for all} e \in \mc{E}} D_{KL}(P||P_{0})\} \hspace{0.1in} \text{s.t.} \hspace{0.1in} \mathbb{E}_P[\mathbb{E}_{P_{0}}[Z|E]] = 0.
\end{align*}
Since $\mathbb{E}_{P_{0}}[Z|E]$ is a function of $E$, using the chain rule for KL divergence,
\begin{align*}
  s_E(\theta,P_0) &=  \exp\{-\min_{P \in \mc{P} :  P(\cdot | E=e) = P_0(\cdot | E=e)  \text{ for all} e \in \mc{E}} D_{KL}(P||P_{0})\} \hspace{0.1in} \text{s.t.} \hspace{0.1in} \mathbb{E}_P[\mathbb{E}_{P_{0}}[Z|E]] = 0\\
  &=\exp\{-\min_{P \in \mc{P} } D_{KL}(P||P_{0})\} \hspace{0.1in} \text{s.t.} \hspace{0.1in} \mathbb{E}_P[\mathbb{E}_{P_{0}}[Z|E]] = 0.
\end{align*}
Now we can use Theorem~\ref{thm:Donsker-Varadhan} for the random variable $\mathbb{E}_{P_{0}}[Z|E]$, which completes the proof.
\end{proof}

\subsection{Proof of Lemma~\ref{lem:consistency-mean-cond}}
\label{proof:consistency-mean-cond}

We will show that for any compact subset $\Lambda \subset \R$,

\begin{equation}
\label{eqn:final-conv-mean}
\sup_{\lambda \in \Lambda}| E_{P_n}[ e^{ \lambda \hat f_n(E)}]-E_{P_0}[ e^{ \lambda \E_{P_0}[Z\mid E]}]| \overset{P}{\to}0.
\end{equation}
Since $E_{P_n}[ e^{ \lambda \hat f_n(E)}]$ and $E_{P_0}[ e^{ \lambda \E_{P_0}[Z\mid E]}]$ are convex functions in $\lambda$, hence, the proof follows from Corollary~\ref{cor:min-conv}.
In order to show \eqref{eqn:final-conv-mean}, it suffices to show the following:

\begin{equation}
\label{eqn:first-ineq-mean}
\sup_{\lambda \in \Lambda}| E_{P_n}[ e^{ \lambda \hat f_n(E)}]-E_{P_n}[ e^{ \lambda \E_{P_0}[Z\mid E]}]| \overset{P}{\to}0 \text{ and }
\end{equation}

\begin{equation}
\label{eqn:second-ineq-mean}
\sup_{\lambda \in \Lambda}| E_{P_n}[ e^{ \lambda \E_{P_0}[Z \mid E]}]-E_{P_0}[ e^{ \lambda \E_{P_0}[Z\mid E]}]| \overset{P}{\to}0.
\end{equation}
\eqref{eqn:second-ineq-mean} follows from Theorem \ref{thm:uniform-convex}. We next show \eqref{eqn:first-ineq-mean}. 
Since $Z$ has finite moment generating function, hence, the random variable $\E_{P_0}[Z \mid E]$ also has finite moment generating function. Hence, for any small $\epsilon >0$, we can choose a large $M \in \R$ such that the set $R=\{e \in \R^d \mid |\E_{P_0}[Z \mid E=e]| \leq M\} $ satisfies $P_0(R) \geq 1-\epsilon$. Now, from Assumption~\ref{ass:conditional-uniform-convergence-mean}, we have 
\begin{equation}
\sup_{\lambda \in \Lambda} \sup_{e \in R}|e^{ \lambda \hat f_n(e)}- e^{ \lambda \E_{P_0}[Z\mid E=e]}| \overset{P}{\to}0.
\end{equation}
%Hence,
%\begin{equation}
%\sup_{\lambda \in \Lambda} \sup_{e \in R}|e^{ \lambda \hat f_n(E)}- e^{ \lambda \E_{P_0}[Z\mid E]}| \overset{P}{\to}0.
%\end{equation}
Hence, 
\begin{equation}
\sup_{\lambda \in \Lambda}| E_{P_n}[ e^{ \lambda \hat f_n(E)} \indic{E \in R}]-E_{P_n}[ e^{ \lambda \E_{P_0}[Z\mid E]} \indic{E \in R}]|  \leq \sup_{\lambda \in \Lambda} \sup_{e \in R}|e^{ \lambda \hat f_n(e)}- e^{ \lambda \E_{P_0}[Z\mid E=e]}| \overset{P}{\to}0.
\end{equation}
Since we can choose the set $R$ with an arbitrarily large probability, we have \eqref{eqn:first-ineq-mean}.

\subsection{Proof of Corollary~\ref{cor:r-values-m}}
\label{proof:r-values-m}
Since $L$ is convex and smooth in its first argument, hence, the minimizer in \eqref{eq:M-est} is equivalently a solution of $E_{P}[\ell(\theta^M, Z)]=0.$ To obtain $s$-value in \eqref{eq:4}, we need to find the distribution $P$ closest to $P_0$ such that $E_{P}[\ell(\eta, Z)]=0.$ Hence, $s$-value in \eqref{eq:4} can be rewritten as

\begin{equation}
   s(\theta^M-\eta,P_0)=\exp\{-\min_{P \in \mc{P}} D_{KL}(P||P_0)\} \hspace{0.1in} \text{s.t.} \hspace{0.1in}   E_{P}[\ell(\eta,Z)]= 0.
\end{equation}
This is the same problem as obtaining $s$-value for a multivariate mean of the random variable $\ell(\eta,Z)$. Hence, following similar arguments as the proof for Theorem \ref{thm:Donsker-Varadhan}, we have the result.

\subsection{Proof of Corollary~\ref{cor:r-values-m-cond}}
\label{proof:r-values-m-cond}

Since $L$ is convex and smooth in its first argument, hence, the minimizer in \eqref{eq:M-est} is equivalently a solution of $E_{P}[\ell(\theta^M, Z)]=0.$ To obtain $s$-value in \eqref{eq:5}, we need to find the distribution $P$ closest to $P_0$ such that $P[\bullet \mid E=e]=P_0[\bullet \mid E=e]$ \text{ for all } $e \in \mc{E}$ and $E_{P}[\ell(\eta, Z)]=0.$ Hence, the $s$-value in \eqref{eq:5} can be rewritten as
\begin{equation}
   s(\theta^M-\eta,P_0)=\exp\{-\min_{P \in \mc{P}, P[\bullet|E=e] = P_{0}[\bullet|E=e] \text{ for all } e \in \mc{E}} D_{KL}(P||P_0)\} \text{ s.t. }  E_{P}[\ell(\eta,Z)]= 0.
\end{equation}
This is the same problem as obtaining directional $s$-value for the multivariate mean of the random variable $\ell(\eta,Z)$. Hence, following similar arguments as the proof for Theorem \ref{thm:Donsker-Varadhan-condn}, we have the result.

\begin{lemma}
\label{lem:gen-rockafellar-10.8}
Let $\mc{X} \subseteq \R^m$ be an open convex set and $\mc{Y} \subset \R^d$ be a compact set. Let $\{f_n\}_{n \geq 1}$ be a sequence of real valued functions defined on $\mc{X} \times \mc{Y}$, where each of the function $f_n$ is convex in the first variable and converges pointwise on $\mc{X} \times \mc{Y}$ to a function $f$, that is 
\begin{align*}
f(x,y)=\lim_{n \to \infty} f_{n}(x,y) \text{ for all } (x,y) \in \mc{X}\times \mc{Y}.
\end{align*}
Suppose that 
\begin{align}
\label{eqn:unif-conv-y}
g_{n}(x)=\sup_{y \in \mc{Y}} |f_{n}(x,y)- f(x,y)| \to 0 \text{ for each } x\in \mc{X} \text{ as } n \to \infty
\end{align}
and 
\begin{align}
\label{eqn:finite-sup}
\sup_{y \in \mc{Y}} |f(x,y)| < \infty \text{ for each } x \in \mc{X}.  
\end{align} 
Then $\sup_{y \in \mc{Y}} |f_{n}(x,y)- f(x,y)| \to 0$ uniformly on each compact  $ S \subset \mc{X}$ as $n \to \infty$.
\end{lemma}
\begin{proof}
The proof works along similar lines as the proof of Theorem 10.8 in \citet{Rockafellar70}. First, we observe that the collection $\{f_n(\cdot, y) \mid n\geq 1 \text{ and } y \in \mc{Y} \}$ is pointwise bounded on $\mc{X}$ using \eqref{eqn:unif-conv-y} and \eqref{eqn:finite-sup}. Hence, by Theorem 10.6 of \citet{Rockafellar70} it is equi-Lipschitzian on each closed bounded subset of $\mc{X}$. Then there exists a real number $\alpha>0$ such that
\begin{align}
|f_{n}(x_1,y)-f_{n}(x_2,y)| \leq \alpha|x_1-x_2|, \text{ for all } x_1,x_2 \in S, n\geq 1 \text{ and } y \in \mc{Y}.
\end{align}

Since $S$ is compact, hence, there exists a finite subset $C_0 $ of $S$ such that each point of $S$ lies within $\frac{\epsilon}{3\alpha}$ distance of at least one point of $C_0$. Since $C_0$ is finite and the functions $g_n$ converge pointwise on $C_0$, there exists an integer $N_0$ such that 
\begin{align}
|f_{n_1}(x,y)-f_{n_2}(x,y)| \leq \frac{\epsilon}{3\alpha} \text{ for all } n_1, n_2 \geq N_0, x \in C_0 \text{ and } y \in \mc{Y}.
\end{align}

Given any $x \in S$, let $z$ be one of the points of $C_0$ such that $|z-x| \leq \frac{\epsilon}{3\alpha}$. Then for all $n_1, n_2 \geq N_0$ and $y \in \mc{Y}$, we have
\begin{align*}
|f_{n_1}(x,y)-f_{n_2}(x,y)| &\leq |f_{n_1}(x,y)-f_{n_1}(z,y)| +|f_{n_1}(z,y)-f_{n_2}(z,y)|+ |f_{n_2}(z,y)-f_{n_2}(x,y)| \\
& \leq \alpha|x-z| +\frac{\epsilon}{3} +\alpha|z-x| \leq \epsilon.
\end{align*}
Hence, the sequence $\{f_n\}_{n \geq 1}$ is cauchy uniformly in $x \in S$ and $y \in \mc{Y}$. Hence, the proof follows.
\end{proof}

\begin{lemma}
\label{lem:unif-conv-convex}
Let $\mc{X} \subseteq \R^m$ be an open convex set and $\mc{Y} \subset \R^d$ be a compact set. Let $\{F_n\}_{n \geq 1}$ be a sequence of real valued random functions defined on $\mc{X} \times \mc{Y}$, where each of the function $F_n$ is convex in the first variable. 

Suppose that 
\begin{align}
\label{eqn:unif-conv-y-prob}
g_{n}(x)&=\sup_{y \in \mc{Y}} |F_{n}(x,y)- f(x,y)| \overset{P}{\to} 0 \text{ for each } x\in \mc{X}  \text{ as } n \to \infty \text{ and }\\
& \sup_{y\in \mc{Y}} |f(x,y)| < \infty \text{ for each } x \in \mc{X}.
\end{align}

Then $\sup_{y \in \mc{Y}} |F_{n}(x,y)- f(x,y)| \overset{P}{\to} 0$ uniformly on each compact $S \subset \mc{X}$ as $ n \to \infty$.
\end{lemma}

\begin{proof}
The proof uses subsequence arguments very similar to that in the proof of Theorem II.1 of \citet{AndersenGi82}. Let $x_1, x_2, \ldots$ be a countable dense set of points in $\mc{X}$. Since $g_n(x_1)\overset{P}{\to} 0$ as $n \to \infty$ there exists a subsequence along which convergence holds almost surely. Along this subsequence $g_n(x_2) \overset{P}{\to} 0$, hence, a further subsequence exists along which $g_n(x_2) \overset{\text{a.s.}}{\to} 0$. By repeating the argument, along a $\text{sub}_k$ sequence, $g_n(x_j) \overset{\text{a.s.}}{\to} 0$ for $j=1, \ldots, k$. By considering the new subsequence formed by taking the first element of the first subsequence, the second element of the second subsequence and so on, we have $g_n(x_j) \overset{\text{a.s.}}{\to} 0$  for each $j=1,2, \ldots$.

Hence, by Lemma \ref{lem:gen-rockafellar-10.8}, it follows that
\begin{align*}
\sup_{x \in S} g_n(x) \overset{\text{a.s.}}{\to}0 \text{ along this subsequence}.
\end{align*}
Since, for any subsequence, there exists a further subsequence along which $\sup_{x \in S} g_n(x) \overset{\text{a.s.}}{\to}0$. It then follows that $\sup_{x \in S} g_n(x) \overset{\text{P}}{\to}0$ along the whole sequence.
\end{proof}

\subsection{Proof of Lemma~\ref{lem:consistency-one-comp}}
\label{proof:consistency-one-comp}
By Assumption~\ref{ass:uniform_convergence}, it follows that $\sup_{\eta \in \Sigma}| E_{P_n}[ e^{ \lambda^{\intercal} \ell(\eta,Z)}]-E_{P_0}[ e^{ \lambda^{\intercal} \ell(\eta,Z)}]| \overset{P}{\to}0$ (see Theorem 19.4 and Example 19.8 of \citet{VanDerVaart02}). Let $\Lambda \subset \R^p$ be a compact subset. Since $e^{ \lambda^{\intercal} \ell(\eta,Z)}$ is convex in $\lambda$, by Assumption~\ref{ass:uniform_convergence} and Lemma~\ref{lem:unif-conv-convex}, we have $\sup_{\lambda \in \Lambda}\sup_{\eta \in \Sigma}| E_{P_n}[ e^{ \lambda^{\intercal} \ell(\eta,Z)}]-E_{P_0}[ e^{ \lambda^{\intercal} \ell(\eta,Z)}]| \overset{P}{\to}0$.

Let $f_n(\lambda, \eta)= E_{P_n}[ e^{ \lambda^{\intercal} \ell(\eta,Z)}]$ and $f(\lambda, \eta)= E_{P_0}[ e^{ \lambda^{\intercal} \ell(\eta,Z)}]$. Since $\sup_{\eta} \sup_{\lambda}|f_n(\lambda, \eta) -f(\lambda,\eta)| \overset{P}{\to} 0$, for any $\epsilon >0$, there exists $N$ such that 
\begin{equation}
\label{eqn:uniform}
|f_n(\lambda, \eta) -f(\lambda,\eta)| <\epsilon \text{ for all } \lambda \in \Lambda, \eta \in \Sigma \text{ and } n>N
\end{equation}
with probability at least $1-\delta$ for some small $\delta>0$.
We first show that $\sup_{\eta \in \Sigma}|\inf_{\lambda \in \Lambda} f_n(\lambda, \eta) - \inf_{\lambda} f(\lambda, \eta) | \overset{P}{\to}0$.

For $\eta \in \Sigma$, let $\lambda_0 \in \Lambda$ be such that $f(\lambda_0, \eta) \leq \inf_{\lambda} f(\lambda, \eta) +\epsilon$. Hence, using \eqref{eqn:uniform}, we have for all $n>N$ 
\begin{align}
\label{eqn:lower}
\inf_{\lambda \in \Lambda} f(\lambda, \eta) &\geq  f(\lambda_0, \eta) -\epsilon  \geq f_n (\lambda_0,\eta) -2\epsilon \geq \inf_{\lambda}f_n(\lambda, \eta)-2\epsilon
\end{align}
with probability at least $1-\delta$. 
Now, for $\eta \in \Sigma$, we choose $\lambda_n \in \Lambda$ such that $f_n(\lambda_n) \leq \inf_{\lambda} f_n(\lambda, \eta) +\epsilon$. Using \eqref{eqn:uniform}, we have
\begin{align}
\label{eqn:upper}
\inf_{\lambda \in \Lambda} f(\lambda, \eta) -\epsilon \leq  f_{n}(\lambda_n, \eta) \leq \inf_{\lambda \in \Lambda} f_{n}(\lambda, \eta) +\epsilon
\end{align}
with probability at least $1-\delta$.

Combining \eqref{eqn:lower} and \eqref{eqn:upper}, we have 
\begin{align}
\label{eqn:bound_mod}
|\inf_{\lambda \in \Lambda} f_n(\lambda, \eta) - \inf_{\lambda \in \Lambda} f(\lambda, \eta)| < 2\epsilon
\end{align}
for all $\eta$ and $n>N$ with probability at least $1-\delta$.

Let $g_{n}( \eta)= \inf_{\lambda \in \Lambda}f_{n}(\lambda, \eta)$ and $g(\eta)=\inf_{\lambda \in \Lambda}f(\lambda, \eta)$, then $\sup_{\eta}|g_n(\eta)-g(\eta)| \overset{P}{\to} 0$. Now we need to show $\sup_{\eta_k}|\sup_{\eta_1,\ldots,\eta_{k-1},\eta_k,\ldots \eta_p} g_n(\eta)-\sup_{\eta_1,\ldots,\eta_{k-1},\eta_k,\ldots \eta_p} g(\eta)| \overset{P}{\to} 0$, which follows similarly as the proof of \eqref{eqn:bound_mod}.

\subsection{Proof of Lemma~\ref{lem:consistency-one-comp-cond}}
\label{proof:consistency-one-comp-cond}
We need to show that for any compact subset $\Lambda \subset \R^p$,
\begin{equation}
\label{eqn:final-conv}
\sup_{\lambda \in \Lambda}\sup_{\eta \in \Sigma}| E_{P_n}[ e^{ \lambda^{\intercal} Q_n(\eta,E)}]-E_{P_0}[ e^{ \lambda^{\intercal} \E_{P_0}[\ell(\eta,Z)\mid E]}]| \overset{P}{\to}0
\end{equation}
and then the rest of the proof follows similarly as in the proof of Lemma~\ref{lem:consistency-one-comp}.
In order to show \eqref{eqn:final-conv}, it suffices to show the following:

\begin{equation}
\label{eqn:first-ineq}
\sup_{\lambda \in \Lambda}\sup_{\eta \in \Sigma}| E_{P_n}[ e^{ \lambda^{\intercal} Q_n(\eta,E)}]-E_{P_n}[ e^{ \lambda^{\intercal} \E_{P_0}[\ell(\eta,Z)\mid E]}]| \overset{P}{\to}0.
\end{equation}

\begin{equation}
\label{eqn:second-ineq}
\sup_{\lambda \in \Lambda}\sup_{\eta \in \Sigma}| E_{P_n}[ e^{ \lambda^{\intercal} \E_{P_0}[\ell(\eta,Z)\mid E]}]-E_{P_0}[ e^{ \lambda^{\intercal} \E_{P_0}[\ell(\eta,Z)\mid E]}]| \overset{P}{\to}0.
\end{equation}
\eqref{eqn:second-ineq} follows similarly as in the proof of Lemma~\ref{lem:consistency-one-comp}.
hence, it remains to show show \eqref{eqn:first-ineq}.

Since $\Lambda$ is a compact set, there exists a real constant $M$ such that $\norm{\lambda}_1 \leq M$ for all $\lambda \in \Lambda$.
Hence,
\begin{align}
\label{eqn:unif-b1}
\begin{split}
| \lambda^{\intercal} Q_n(\eta,e)-\lambda^{\intercal} \E_{P_0}[\ell(\eta,Z)\mid E=e]| &\leq \norm{\lambda}_1 \norm{Q_n(\eta,e)-\E_{P_0}[\ell(\eta,Z)\mid E=e]}_{\infty} \\ &\leq M \norm{Q_n(\eta,e)-\E_{P_0}[\ell(\eta,Z)\mid E=e]}_{\infty}.
\end{split}
\end{align}
Now, for any fixed value of $E=e$, for some $\psi_{n,\lambda,e} \in \R$ such that $\lambda^{\intercal} Q_n(\eta,e) \leq \psi_{n,\lambda,e} \leq \lambda^{\intercal} \E_{P_0}[\ell(\eta,Z)\mid E=e]$, we have by Taylor's expansion
\begin{align}
\begin{split}
&|e^{ \lambda^{\intercal} Q_n(\eta,e)}-e^{ \lambda^{\intercal} \E_{P_0}[\ell(\eta,Z)\mid E=e]}|\leq e^{  \psi_{n,\lambda,e}} | \lambda^{\intercal} Q_n(\eta,e)- \lambda^{\intercal} \E_{P_0}[\ell(\eta,Z)\mid E=e]| \\
& \leq e^{ \lambda^{\intercal} \E_{P_0}[\ell(\eta,Z)\mid E=e] } e^{ M \norm{Q_n(\eta,e)-\E_{P_0}[\ell(\eta,Z)\mid E=e]}_{\infty}}M \norm{Q_n(\eta,e)-\E_{P_0}[\ell(\eta,Z)\mid E=e]}_{\infty}
\end{split}
\end{align}
where the last inequality follows from \eqref{eqn:unif-b1}. 

Since by Assumption~\ref{ass:conditional-uniform-convergence}, we have  $\sup_{\eta} \sup_{e} \norm{ E_{P_{0}}[\ell(\eta,Z)|E=e] - Q_n(\eta,e)}_{\infty} \rightarrow 0$. Hence, in order to show \eqref{eqn:first-ineq}, it suffices to show that $\sup_{\lambda \in \Lambda}\sup_{\eta \in \Sigma}\E_{P_n}[e^{ \lambda^{\intercal} \E_{P_0}[\ell(\eta,Z)\mid E] } ] \leq C <\infty$ for some numerical constant $C$ independent of $n$ with high probability. Now, by Assumption~\ref{ass:uniform_convergence}, we have $\E_{P_0}[\sup_{\eta \in \Sigma}e^{\lambda^{\intercal}\ell(\eta,Z)}] < \infty$ for any $\lambda \in \R^p$. Hence, by Jensen's inequality, we have $\E_{P_0}[\sup_{\eta \in \Sigma}e^{\lambda^{\intercal}\E_{P_0}[\ell(\eta,Z)\mid E]}] < \infty$ for any $\lambda \in \R^p$. Also, $\lambda \to \sup_{\eta \in \Sigma}e^{\lambda^{\intercal}\E_{P_0}[\ell(\eta,Z)\mid E=e]}$ is a convex function for any $e$.  Hence, by Theorem~\ref{thm:uniform-convex}, we have the result.

 \section{Additional experimental details}
 \label{sec:appendix-expt}

\begin{table}[h]
  \caption{National supported work demonstration (NSW) data. Table showing the sample means of covariates for the subset extracted by \citet{DehejiaWa99} (DJW) and the subset containing the remaining samples (DJWC) along with $p$-values for testing the difference of means between the two treated and control groups respectively using Welch two sample t-test. 
\newline
Age=age in years; Education=number of years of schooling; Black=1 if black, 0 otherwise; Hispanic=1 if Hispanic, 0 otherwise; Nodegree=1 if no high school degree, 0 otherwise; Married =1 if married, 0 otherwise; RE75= Earnings in 1975.}
 \label{tab:lalonde}
  \centering
  \begin{tabular}{lllllllll}
    \toprule
    \multicolumn{2}{c}{Part}                   \\
    \cmidrule(r){1-2}
    &No. of Obs.& Age &Education&Black&Hispanic\\
    \midrule
 Treated (DJW)   & 185  & 25.81 &10.35& 0.84 &0.06 \\

Treated (DJWC)& 112 &22.66 &10.44    &0.73&0.15\\
 \midrule
p-values for diff. in means&&\textbf{1.95e-05}&0.6501&\textbf{0.027}&\textbf{0.017} \\
\midrule
Control (DJW)& 260 &  25.05& 10.09   &0.83&0.11\\

Control (DJWC)& 165 & 23.49 &10.35   &0.76&0.12\\

\midrule
p-values for diff. in means&&\textbf{0.0123}&0.1111&\textbf{0.09}&0.6724\\
    \bottomrule
  \end{tabular}
\end{table}

\begin{table}[h]
  \caption{National supported work demonstration (NSW) data. Table showing the sample means of covariates for the subset extracted by \citet{DehejiaWa99} (DJW) and the subset containing the remaining samples (DJWC) along with $p$-values for testing the difference of means between the two treated and control groups respectively using Welch two sample t-test. 
\newline
Age=age in years; Education=number of years of schooling; Black=1 if black, 0 otherwise; Hispanic=1 if Hispanic, 0 otherwise; Nodegree=1 if no high school degree, 0 otherwise; Married =1 if married, 0 otherwise; RE75= Earnings in 1975.}
 \label{tab:lalonde}
  \centering
  \begin{tabular}{lllll}
    \toprule
    \multicolumn{2}{c}{Part}                   \\
    \cmidrule(r){1-2}
    &No. of Obs.& Nodegree&Married&RE75\\
    \midrule
 Treated (DJW)   & 185  & 0.19&0.71& 1532.1\\

Treated (DJWC)& 112 &0.13&0.77&5600\\
 \midrule
p-values for diff. in means&&0.2035&0.2539&\textbf{5.42e-10} \\
\midrule
Control (DJW)& 260 &0.15&0.83&1266.9\\

Control (DJWC)& 165 &0.16 &0.78 &5799.66\\

\midrule
p-values for diff. in means&&0.7891&0.1841& \textbf{6.967e-15}\\
    \bottomrule
  \end{tabular}
\end{table}

\end{document}